\documentclass{article}
\usepackage{graphicx}
\usepackage{multirow}

\usepackage[T1]{fontenc}
\usepackage[sc]{mathpazo}
\usepackage{amsmath}
\usepackage{amssymb}
\usepackage{enumerate}
\usepackage{amsthm}
\usepackage{amsfonts,mathrsfs}
\usepackage[vcentermath]{youngtab}

\usepackage{hyperref}
\hypersetup{pdfpagemode=UseNone}

\newcommand{\tinyspace}{\mspace{1mu}}

\newcommand{\op}[1]{\operatorname{#1}}

\newcommand{\abs}[1]{\left\lvert\tinyspace #1 \tinyspace\right\rvert}

\newcommand{\norm}[1]{\left\lVert\tinyspace #1 \tinyspace\right\rVert}

\newcommand{\setft}[1]{\mathrm{#1}}

\newcommand{\density}[1]{\setft{D}\left(#1\right)}
\newcommand{\unitary}[1]{\setft{U}\left(#1\right)}

\newcommand{\spn}{\op{span}}

\def\dif{\mathrm{d}}

\def\complex{\mathbb{C}}
\def\real{\mathbb{R}}
\def\natural{\mathbb{N}}

\def\I{\mathbb{1}}

\newenvironment{mylist}[1]{\begin{list}{}{
    \setlength{\leftmargin}{#1}
    \setlength{\rightmargin}{0mm}
    \setlength{\labelsep}{2mm}
    \setlength{\labelwidth}{8mm}
    \setlength{\itemsep}{0mm}}}
    {\end{list}}

\def\ot{\otimes}

\newcommand{\out}[2]{| #1\rangle\langle #2 |}

\newcommand{\Innerm}[3]{\left\langle #1 \left| #2 \right| #3 \right\rangle}

\newcommand{\Pa}[1]{\left(#1\right)}

\newcommand{\Br}[1]{\left[#1\right]}
\newcommand{\set}[1]{\{#1\}}
\newcommand{\Set}[1]{\left\{#1\right\}}

\newcommand{\ket}[1]{|#1\rangle}

\DeclareMathOperator{\trace}{Tr}

\newcommand{\Ptr}[2]{\trace_{#1}\Pa{#2}}

\newcommand{\Tr}[1]{\Ptr{}{#1}}

\def\cA{\mathcal{A}}\def\cB{\mathcal{B}}\def\cE{\mathcal{E}}
\def\cH{\mathcal{H}}

\def\cU{\mathcal{U}}

\def\bP{\mathbf{P}}\def\bQ{\mathbf{Q}}

\def\rH{\mathrm{H}}

\def\rS{\mathrm{S}}
\def\rU{\mathrm{U}}

\def\sE{\mathscr{E}}

\def\sL{\mathscr{L}}

\newtheorem{thrm}{Theorem}[section]
\newtheorem{lem}[thrm]{Lemma}
\newtheorem{prop}[thrm]{Proposition}
\newtheorem{cor}[thrm]{Corollary}
\theoremstyle{definition}

\newtheorem{remark}[thrm]{Remark}

\numberwithin{equation}{section}

\newcounter{questionnumber}

\newcommand{\aop}{Ann. Phys.~}

\newcommand{\jmp}{J. Math. Phys.~}
\newcommand{\jpa}{J. Phys. A~}

\newcommand{\natphy}{Nature Phys.~}
\newcommand{\natcom}{Nature Commun.~}

\newcommand{\njp}{New. J. Phys.~}

\newcommand{\prl}{Phys. Rev. Lett.~}
\newcommand{\pra}{Phys. Rev. A~}
\newcommand{\pre}{Phys. Rev. E~}

\newcommand{\rmp}{Rev. Math. Phys.~}

\begin{document}

\title{Average of uncertainty-product for bounded observables}

\author{Lin Zhang$^1$\footnote{E-mail: linyz@zju.edu.cn; godyalin@163.com},\quad Jiamei Wang$^2$
\footnote{E-mail: wangjm@ahut.edu.cn}\\
  {\it\small $^1$Institute of Mathematics, Hangzhou Dianzi University, Hangzhou 310018, PR~China}\\
  {\it \small $^2$Department of Mathematics, Anhui University of
Technology, Ma Anshan 243032, PR China}}
\date{}
\maketitle
\maketitle \mbox{}\hrule\mbox\\
\begin{abstract}

The goal of this paper is to calculate exactly the average of
uncertainty-product of two bounded observables and to establish its
typicality over the whole set of finite dimensional quantum pure
states. Here we use the uniform ensembles of pure and isospectral
states as well as the states distributed uniformly according to the
measure induced by the Hilbert-Schmidt norm. Firstly, we investigate
the average uncertainty of an observable over isospectral density
matrices. By letting the isospectral density matrices be of
rank-one, we get the average uncertainty of an observable restricted
to pure quantum states. These results can help us check how large
the gap is between the uncertainty-product and any obtained lower
bounds about the uncertainty-product. Although our method in the
present paper cannot give a tighter lower bound of
uncertainty-product for bounded observables, it can help us drop any
one that is not tighter than the known one substantially.
\end{abstract}
\maketitle \mbox{}\hrule\mbox

\section{Introduction}

Uncertainty principle (aka Heisenberg's uncertainty relation) is one
of basic constraints in quantum mechanics. It means that we cannot
principally obtain precise measurement outcomes simultaneously when
we measure two incomparable observables at the same time. The
mathematical formulation of uncertainty relation is in terms of any
of a variety of inequalities, where a fundamental limit to the
precision with which certain pairs of physical properties of a
particle, i.e. complementary variables, such as position $\hat x$
and momentum $\hat p$, can be known simultaneously. The uncertainty
relation \cite{heisenberg1927zp}, introduced by Heisenberg in 1927,
relates the standard deviation of momentum $\Delta \hat p$ and the
standard deviation of position $\Delta \hat x$, it indicates that
the more precisely the momentum of some particle is determined, the
less precisely its position can be known, and vice versa.
Specifically, the quantitative relation of such two standard
deviations was derived by Kennard \cite{Kennard1927} later that
year:
\begin{eqnarray}
\Delta\hat x\cdot\Delta\hat p\geqslant \frac{\hbar}2,
\end{eqnarray}
where $\Delta\hat p = \sqrt{\langle \hat p^2\rangle - \langle \hat
p\rangle^2}$ and $\Delta\hat x = \sqrt{\langle \hat x^2\rangle -
\langle \hat x\rangle^2}$.

The most common general form of the uncertainty principle is the
Robertson-Schr\"{o}dinger uncertainty relations
\cite{robertson1929pr,schrodinger1930}. In order to state it
explicitly, we need some notions. The precision to which the value
of an observable $A$ can be known is quantified by its uncertainty
function
\begin{eqnarray}
\Delta A(\rho):=\sqrt{\langle A^2\rangle_\rho - \langle
A\rangle^2_\rho}
\end{eqnarray}
where $\langle O\rangle_\rho:=\Tr{O\rho}$ for any observable $O$.
Furthermore, the precision to which the values of two observables
$A$ and $B$ can be known simultaneously is limited by the
Robertson-Schr\"{o}dinger uncertainty relation
\begin{eqnarray}\label{eq:R-S}
\Pa{\Delta A(\rho) \cdot \Delta B(\rho)}^2\geqslant
\Pa{\langle\set{A,B}\rangle_\rho-\langle A\rangle_\rho\langle
B\rangle_\rho}^2 + \langle [A,B]\rangle^2_\rho,
\end{eqnarray}
where $\set{A,B}:=\frac12(AB+BA)$ and
$[A,B]:=\frac1{2\sqrt{-1}}(AB-BA)$. We see from
Robertson-Schr\"{o}dinger uncertainty relation that this uncertainty
relation depends on the state under consideration. There are a lot
of literatures devoting to improve the right hand side (rhs) of the
above inequality \cite{Berta2010nap,Salimi2016}. Moreover, recently
many researchers proposed new perspective, instead of description of
uncertainty-product, they used the sum of uncertainty
\cite{Huang2012,maccone2014}, and its various generalizations
\cite{chen2016qip,qin2016scr}, etc. Besides, many researchers
generalize the uncertainty relation from pure state to isospectral
mixed states by employing symplectic geometric tools
\cite{Andersson2014jmp}. Many contributions are given to another
reformulation of uncertainty relation, for instance entropic
uncertainty relation \cite{tomamichel2011prl,wehner2010njp} and its
applications \cite{coles2015arxiv}. A connection is also established
between entropic uncertainty and wave-particle duality
\cite{coles2014natcom}. There are literatures devoted to study the
connection among uncertainty, and entanglement
\cite{berta2014pra,guhne2004pra,Huang2010,Huang2012}, and the
reversibility of measurement \cite{berta2015arxiv}.

The purpose of this paper is to give a new perspective to
state-independent uncertainty relation in terms of representation
theory of unitary group and random matrix theory. Caution: because
observables may be unbounded, for instance, the position operator
$\hat x$, in physical regime, an unbounded observable may take
infinity at some state. Throughout this paper, we will focus on
bounded observables. Consider the following particular statistical
ensembles: The used distribution of random state is uniform
distribution induced by Hilbert-Schmidt measure defined over the set
of all density matrices. By using tools from representation theory
of unitary group and random matrix theory, we can give an exact
calculation of such average value (in the pure state case or mixed
state case, respectively) and consider its typicality under some
restriction. Theoretically, as the typicality suggests that without
measuring such bounded observables, we may claim that at most
sampled states, one can get their uncertainty-product is close to
their average value with overwhelming probability. Equivalently,
their uncertainty-product deviates their average value with
exponentially small probability. Our method proposed here in fact
can help check how large the gap is between the uncertainty-product
and any obtained lower bounds about the uncertainty-product.
Specifically, except calculate the average of uncertainty-product,
we also calculate the averages of the obtained lower bounds of
uncertainty-product. Clearly the obtained lower bounds are
state-dependent.

This paper is organized as follows. In
Sect.~\ref{sect:metric-on-state-space}, we will introduce various
measures on state space. Specifically, there is a unique probability
which is unitarly invariant on the pure state space. But, however,
there is no unique unitarily-invariant probability measure over the
mixed state space because of the existence of environment.
Sect.~\ref{sect:motivation} discusses the motivation why we take the
average over corresponding state ensembles. Sect.~\ref{sect:iso-ave}
deals with the isospectral average of uncertainty-product of two
bounded observables over the set of isospectral quantum states.
Furthermore, separately, we consider the average of
uncertainty-product for a random pure state, and also for a random
mixed state. In Sect.~\ref{sect:concentration}, we make a discussion
about the concentration of measure phenomenon about the quantity,
i.e., the uncertainty-product of two bounded observables over the
set of mixed states. Finally, some necessary materials for reasoning
of our results are provided in the Appendix, see
Sect.~\ref{sect:appendix}, for example, two specific examples in
lower dimensions are provided in Sect.~\ref{sect:examples}.

\section{Measures on the state
spaces}\label{sect:metric-on-state-space}

Given a measure $\mu$ on the set of quantum states, one can
calculate the corresponding averages over all states with respect to
this measure \cite{Zyczkowski2001jpa}. We will consider the set of
pure quantum states. For a $d$-dimensional Hilbert space $\cH_d$,
the set of pure states consists of all unit vectors in $\cH_d$. On
this set, there exists a unique measure which is unitarily
invariant, i.e., uniform probability measure $\dif \mu(\psi)$ or
induced by normalized Haar measure $\dif\mu_{\mathrm{Haar}}(U)$ over
the unitary group $\rU(d)$. Indeed, any random pure state
$\ket{\psi}$ is generated by a random unitary matrix $U\in\rU(d)$ on
any fixed pure state $\ket{\psi_0}$ via $\ket{\psi}=U\ket{\psi_0}$.
The uniform ensemble of pure quantum states of finite-dimensional
Hilbert space studied extensively in the context of foundations of
quantum statistical mechanics, entanglement theory or various
protocols/features of quantum information theory. Related
literatures are too numerous to mention. Here we mention our two
works using such particular ensemble to investigate the typicality
of quantum coherence and average entropy of isospectral quantum
states, see \cite{uttam2016pra,zhang2015arxiv}. Then we can define
the average value of some function $f$ on the set of pure states as
follows:
\begin{eqnarray}
\langle f(\psi)\rangle := \int_{\mathbb{S}^k}f(\psi)\dif \mu(\psi) =
\int_{\rU(d)} f(U\psi_0)\dif\mu_{\mathrm{Haar}}(U).
\end{eqnarray}

Unlike the case of pure states, it is known that there exist various
measures on the set of mixed states, $\density{\cH_d}$, the set of
all positive semidefinite matrices with unit trace. As a matter of
fact, one assumes naturally the distributions of eigenvalues and
eigenvectors of a quantum state $\rho$, via the spectral
decomposition $\rho=U\Lambda U^\dagger$, are independent. Thus any
probability measure $\mu$ on $\density{\cH_d}$ will be of product
form: $\dif\mu(\rho) = \dif\nu(\Lambda)\times
\dif\mu_{\mathrm{Haar}}(U)$, where $\dif\mu_{\mathrm{Haar}}(U)$ is
the unique Haar measure \cite{gessner2013pre} on the unitary group
and $\nu$ defines the distribution of eigenvalues without unique
choice for it. The utility of $\nu$ in the average entropy or
average coherence can be found in
\cite{Zhang2016,zhang2016arxiv,Zyczkowski2001jpa}.

The measures used frequently over the $\density{\cH_d}$ can be
obtained by partially tracing over the Haar-distributed pure states
in the higher dimension Hilbert space $\cH_d\ot \cH_k$, say
$\complex^d\ot\complex^k$. In order to be convenience we suppose
that $d\leqslant k$. Following \cite{Zyczkowski2001jpa}, the joint
probability density function of spectrum
$\Lambda=\Set{\lambda_1,\ldots,\lambda_d}$ of $\rho$ is given by
\begin{eqnarray}
\dif\nu_{d,k} (\Lambda) =
C_{d,k}\delta\Pa{1-\sum^d_{j=1}\lambda_j}\prod_{1\leqslant
i<j\leqslant
 d}(\lambda_j-\lambda_i)^2\prod^d_{j=1}\lambda^{k-d}_j\theta(\lambda_j)\dif\lambda_j,
\end{eqnarray}
where the theta function $\theta$ ensures that $\rho$ is positive
definite, $C_{d,k}$ is the normalization constant, given by
\begin{eqnarray}
C_{d,k} =
\frac{\Gamma(dk)}{\prod^{d-1}_{j=0}\Gamma(d-j+1)\Gamma(k-j)}.
\end{eqnarray}
In particular, in the present paper we will consider a special case
where $d=k$, which corresponds to the Hilbert-Schmidt measure, a
flat metric over the $\density{\cH_d}$, denoted by
$\dif\mu_{\rH\rS}(\rho)$. We also denote $\dif \nu_{d,k}=\dif \nu$
and $C_{d,k}=C_{\rH\rS}$ if $d=k$. Thus we have
\begin{eqnarray}
\dif\mu_{\rH\rS}(\rho) = \dif\nu(\Lambda)\times
\dif\mu_{\mathrm{Haar}}(U),
\end{eqnarray}
where $\rho=U\Lambda U^\dagger$.

For convenience, we let $A, B$ be observables,  $\rho=U\Lambda
U^\dagger$, and introduce the following symbol for convenience:
\begin{eqnarray}\label{eq:expect-UU}
t_k=\Tr{\Lambda^k}=\Tr{\rho^k},\quad\sE_k(\Lambda):= \int (U\Lambda
U^\dagger)^{\ot k} \dif\mu_{\mathrm{Haar}}(U).
\end{eqnarray}

\section{Motivation}\label{sect:motivation}

In order to explain why we take the average of uncertainty-product
for bounded observables, some words are needed. Denote
$L_0(A,B,\rho):=\Pa{\langle\set{A,B}\rangle_\rho-\langle
A\rangle_\rho\langle B\rangle_\rho}^2 + \langle
[A,B]\rangle^2_\rho$. Clearly Eq.~\eqref{eq:R-S} becomes $\Pa{\Delta
A(\rho) \cdot \Delta B(\rho)}^2\geqslant L_0(A,B,\rho)$. If one is
obtained another lower bound, say $L(A,B,\rho)$, via some
mathematical methods, then $\Pa{\Delta A(\rho) \cdot \Delta
B(\rho)}^2\geqslant L(A,B,\rho)$. Now we need to compare lower
bounds $L_0$ and $L$. If $L(A,B,\rho)\geqslant L_0(A,B,\rho)$, then
we can say the lower bounds of the uncertainty principle are
improved, that is, we get a tighter lower bound. However, such
improvement sometimes is not essential, it is possible that
\begin{eqnarray}
\int L(A,B,\rho)\dif\mu_{\rH\rS}(\rho)= \int
L_0(A,B,\rho)\dif\mu_{\rH\rS}(\rho).
\end{eqnarray}
This shows that $L(A,B,\rho)=L_0(A,B,\rho)$ is satisfied almost
every except a zero measure in state space by the Measure Theory.
This is not real improvement. In fact, there are two observables
such that the lower bound of the uncertainty principle
Eq.~\eqref{eq:R-S} cannot be improved, see below Eq.~\eqref{eq:AAA}.
This example tell us that getting a universal uncertainty principle
for any observables in which the lower bound is really improved,
compared with $L_0$, seems impossible. At least, the statement is
applicable for Eq.~\eqref{eq:AAA}. However, if $\Pa{\Delta A(\rho)
\cdot \Delta B(\rho)}^2\geqslant L(A,B,\rho)\geqslant L_0(A,B,\rho)$
and
\begin{eqnarray}
\int L(A,B,\rho)\dif\mu_{\rH\rS}(\rho)> \int
L_0(A,B,\rho)\dif\mu_{\rH\rS}(\rho),
\end{eqnarray}
then we say that the uncertainty principle $\Pa{\Delta A(\rho) \cdot
\Delta B(\rho)}^2\geqslant L(A,B,\rho)$ really improves the
uncertainty principle $\Pa{\Delta A(\rho) \cdot \Delta
B(\rho)}^2\geqslant L_0(A,B,\rho)$. Therefore $L(A,B,\rho)$ is
tighter than $L_0(A,B,\rho)$ substantially. This is what we want.
But there is another situation that appears. We maybe get a new one
$\hat L(A,B,\rho)$ without knowing the relationship between $\hat L$
and $L_0$. But we can still determine wether or not
\begin{eqnarray}
\int \hat L(A,B,\rho)\dif\mu_{\rH\rS}(\rho)> \int
L_0(A,B,\rho)\dif\mu_{\rH\rS}(\rho).
\end{eqnarray}
If it were the case, then $\hat L(A,B,\rho)> L_0(A,B,\rho)$ would
hold in a subset of the state space. Improvement of uncertainty
principle is possible limited to local range.

\section{Isospectral average of uncertainty-product}\label{sect:iso-ave}

In this section, we focus on the ensemble of isospectral density
matrices. This ensemble has been recently studied in various
contexts of quantum information. In fact, we also do some work in
this field \cite{zhang2015arxiv}.

Consider the set of all isospectral density matrices
$\cU_\Lambda:=\set{\rho:\rho=U\Lambda U^\dagger, U\in\unitary{d}}$
with a fixed spectrum $\Lambda=\set{\lambda_1,\ldots,\lambda_d}$,
where $\lambda_j\geqslant0$ for each $j$ and
$\sum^d_{j=1}\lambda_j=1$. Now we can explicitly compute the average
(squared) uncertainty of observable $A$ over the set of isospectral
density matrices $\cU_\Lambda$ as follows:
\begin{eqnarray}
&&\int\Delta A(\rho)^2
\dif\mu_{\mathrm{Haar}}(U)=\Tr{A^2\sE_1(\Lambda)} - \Tr{A^{\ot
2}\sE_2(\Lambda)},
\end{eqnarray}
where $\cE_k(\Lambda)$ is from \eqref{eq:expect-UU}. The details of
computation about $\cE_k(\Lambda)$, where $k=1,2,3,4$, are gathered
in the Appendix, i.e., Section~\ref{sect:appendix}.

From the relations~\eqref{e1} and ~\eqref{e2}, we see that
\begin{eqnarray}
\int\Delta A(\rho)^2 \dif\mu_{\mathrm{Haar}}(U) =
\frac{d-\Tr{\Lambda^2}}{d^2-1}\Br{\Tr{A^2} - \frac1d(\Tr{A})^2}.
\end{eqnarray}
By ~\eqref{e7}, we have
\begin{eqnarray}
\int_{\density{\cH_d}}\Delta A(\rho)^2 \dif\mu_{\rH\rS}(\rho) =
\frac d{d^2+1}\Br{\Tr{A^2} - \frac1d(\Tr{A})^2}.
\end{eqnarray}
On the other hand, for any state $\rho\in\density{\cH_d}$,
\begin{eqnarray}\label{eq:prod-uncertainty}
\Delta A(\rho)^2\cdot\Delta B(\rho)^2 &=& \Tr{\Br{A^2\ot
B^2}\rho^{\ot 2}} +
\Tr{\Br{A^{\ot 2}\ot B^{\ot 2}}\rho^{\ot 4}} \notag\\
&&- \Tr{\Br{A^2\ot B^{\ot 2}}\rho^{\ot 3}} - \Tr{\Br{B^2\ot A^{\ot
2}}\rho^{\ot 3}}.
\end{eqnarray}
Thus
\begin{eqnarray}
&&\int\Delta A(\rho)^2\cdot\Delta B(\rho)^2\dif\mu_{\mathrm{Haar}}(U) \notag\\
&&= \Tr{\Br{A^2\ot
B^2}\sE_2(\Lambda)} +
\Tr{\Br{A^{\ot 2}\ot B^{\ot 2}}\sE_4(\Lambda)} \notag\\
\label{6}&&~~~~- \Tr{\Br{A^2\ot B^{\ot 2}}\sE_3(\Lambda)} - \Tr{\Br{B^2\ot
A^{\ot 2}}\sE_3(\Lambda)},
\end{eqnarray}
where $\rho\in \cU_\Lambda$.

With these identities, we calculate the the averaged
uncertainty-product over the isospectral density matrices. By the
tedious but simple calculations, we have the following result:
\begin{thrm}\label{th:iso}
For two observables $A$ and $B$ on $\cH_d$, the average of
uncertainty-product over the set of all isospectral density matrices $\rho$
on $\cH_d$ is given by a symmetric function in arguments $A$ and $B$
\begin{eqnarray}\label{eq:8-symmetric}
\int\Delta A(\rho)^2\cdot\Delta B(\rho)^2\dif\mu_{\mathrm{Haar}}(U) =
\sum^8_{j=1}\omega_j(\Lambda)\cdot\Omega_j(A,B),
\end{eqnarray}
where $\Omega_j(A,B)$ are symmetric in arguments $A$ and $B$ for
each $j$: $\Omega_j(A,B)=\Omega_j(B,A)$ and
\begin{eqnarray}
\Omega_1(A,B)&=& \Tr{A}^2\Tr{B}^2,\\
\Omega_2(A,B)&=& \Tr{A^2}\Tr{B}^2 +\Tr{A}^2\Tr{B^2},\\
\Omega_3(A,B)&=& \Tr{AB}\Tr{A}\Tr{B}, \\
\Omega_4(A,B)&=& \Tr{A^2}\Tr{B^2}, \\
\Omega_5(A,B)&=& \Tr{AB}\Tr{AB}, \\
\Omega_6(A,B)&=& \Tr{A^2B}\Tr{B}+ \Tr{A}\Tr{AB^2},\\
\Omega_7(A,B)&=& \Tr{A^2B^2},\\
\Omega_8(A,B)&=&\Tr{ABAB},
\end{eqnarray}
and $\omega_j(\Lambda)$ are given by the following:
\begin{eqnarray}
\omega_1(\Lambda) &=& \frac1{24}\Delta^{(4)}_4 + \frac38\Delta^{(3,1)}_4 + \frac16\Delta^{(2,2)}_4 + \frac38\Delta^{(2,1,1)}_4 + \frac1{24}\Delta^{(1,1,1,1)}_4,\\
\omega_2(\Lambda) &=& \Pa{\frac1{24}\Delta^{(4)}_4 + \frac18\Delta^{(3,1)}_4 - \frac18\Delta^{(2,1,1)}_4 - \frac1{24}\Delta^{(1,1,1,1)}_4} - \Pa{\frac16\Delta^{(3)}_3+\frac23\Delta^{(2,1)}_3+\frac16\Delta^{(1,1,1)}_3},\\
\omega_3(\Lambda) &=& \frac16\Delta^{(4)}_4 + \frac12\Delta^{(3,1)}_4 - \frac12\Delta^{(2,1,1)}_4 - \frac16\Delta^{(1,1,1,1)}_4,\\
\omega_4(\Lambda) &=&\Pa{\frac1{24}\Delta^{(4)}_4 - \frac18\Delta^{(3,1)}_4 + \frac16\Delta^{(2,2)}_4 - \frac18\Delta^{(2,1,1)}_4 + \frac1{24}\Delta^{(1,1,1,1)}_4}  \notag\\
&&+\Pa{\frac{\Delta^{(2)}_2}2
+ \frac{\Delta^{(1,1)}_2}2} - \Pa{\frac13\Delta^{(3)}_3-\frac13\Delta^{(1,1,1)}_3},\\
\omega_5(\Lambda) &=& \frac1{12}\Delta^{(4)}_4 - \frac14\Delta^{(3,1)}_4 + \frac13\Delta^{(2,2)}_4 - \frac14\Delta^{(2,1,1)}_4 + \frac1{12}\Delta^{(1,1,1,1)}_4,\\
\omega_6(\Lambda) &=& \Pa{\frac16\Delta^{(4)}_4 - \frac13\Delta^{(2,2)}_4 + \frac16\Delta^{(1,1,1,1)}_4} - \Pa{\frac13\Delta^{(3)}_3-\frac13\Delta^{(1,1,1)}_3},
\end{eqnarray}
\begin{eqnarray}
\omega_7(\Lambda) &=& \Pa{\frac16\Delta^{(4)}_4 - \frac12\Delta^{(3,1)}_4 + \frac12\Delta^{(2,1,1)}_4 - \frac16\Delta^{(1,1,1,1)}_4} \notag\\
&&+ \Pa{\frac{\Delta^{(2)}_2}2 - \frac{\Delta^{(1,1)}_2}2} - 2\Pa{\frac13\Delta^{(3)}_3 - \frac23\Delta^{(2,1)}_3+\frac13\Delta^{(1,1,1)}_3},\\
\omega_8(\Lambda) &=& \frac1{12}\Delta^{(4)}_4 -
\frac14\Delta^{(3,1)}_4 + \frac14\Delta^{(2,1,1)}_4 -
\frac1{12}\Delta^{(1,1,1,1)}_4.
\end{eqnarray}
Here the meanings of the notations
$\Delta^{(4)}_4,\Delta^{(3,1)}_4,\Delta^{(2,1,1)}_4,\Delta^{(1,1,1,1)}_4$
can be found from \eqref{eq:delta4} to \eqref{eq:delta1111}.
\end{thrm}
The hard part of the proof centers around the calculations of
$\cE_k(\Lambda)$ by using Schur-Weyl duality. Among other things,
the key ingredient here is the Weingarten function, defined over the
permutation group $S_k$, see the definition \eqref{eq:weingarten}
for the unitary group. There are many ways that can be used to
define the Weingarten function, for instance, a sum over partitions
or equivalently, Young tableaux of $k\in\natural$ and the characters
of the symmetric group. In the case where permutation groups of
lower orders are considered (such as $k=2,3,4$ in our paper), the
Weingarten functions can be explicitly evaluated. When $k$ becomes
larger, the explicit evaluation of such function is considerably
complicated, and naturally the asymptotics is concerned. The proof
of Theorem~\ref{th:iso} is placed in Section~\ref{subsect:iso}.

\begin{remark}
Let $N^{-1}_d= d^2(d^2-1)(d^2-4)(d^2-9)$. We can write down more
specific expressions for $\omega_j(\Lambda)$, where
$t_k=\Tr{\Lambda^k}$ for natural number $k$.
\begin{eqnarray}
\omega_1(\Lambda) &=& N_d\Pa{(d^4-8d^2+6) -6d(d^2-4)t_2 +
3(d^2+6)t^2_2 +
8(2d^2-3)t_3 -30dt_4},\\
\omega_2(\Lambda) &=& N_d\Pa{-d(d^4-10d^2+14) + 2d^2(2d^2-13)t_2
-d(d^2+6)t^2_2-8d(d^2-4)t_3 + 10d^2t_4},\\
\omega_3(\Lambda) &=& N_d\Pa{-4d(d^2-4) + 4d^2(d^2+1)t_2 -4d(d^2+6)t^2_2
-16d(d^2+1)t_3 + 40d^2t_4},\\
\omega_4(\Lambda) &=& N_d\left((d^6-11d^4+19d^2+6)-d(3d^4-25d^2+12)t_2\right.
\\&&\left.+
(d^4-6d^2+18)t^2_2 + 4(d^4-5d^2-6)t_3 -2d(2d^2-3)t_4\right),\\
\omega_5(\Lambda) &=& N_d\Pa{2(d^2+6) -4d(d^2+6)t_2 +
2(d^4-6d^2+18)t^2_2
+ 16(2d^2-3)t_3-4d(2d^2-3)t_4},\\
\omega_6(\Lambda) &=& N_d\left(2(d^4-5d^2-6)-2d(d^4-d^2-12)t_2+
12(2d^2-3)t^2_2\right.\\&& \left.+ 8(d^4-3d^2+6)t_3 -12d(d^2+1)t_4\right),\\
\omega_7(\Lambda) &=& N_d\left(-d(d^3+4d^2-9d-16) +
d^2(d^4-d^2-32)t_2 \right.\\&&\left.-4d(2d^2-3)t^2_2 -4d(d^4-5d^2+4)t_3 + 4d^2(d^2+1)t_4\right),\\
\omega_8(\Lambda) &=& N_d\Pa{-10d + 20d^2t_2-2d(2d^2-3)t^2_2 -8d(d^2+1)t_3 + 2d^2(d^2+1)t_4}.
\end{eqnarray}
Because $A$ and $B$ are bounded observables, i.e., Hermitian
operators, we see that $\Tr{A},\Tr{B}$, and $\Tr{AB}$ are real
numbers and $\Tr{A^2}\geqslant0$, and $\Tr{B^2}\geqslant0$. Then
$\Tr{A}^2\Tr{B}^2\geqslant0,\Tr{A^2}\Tr{B}^2
+\Tr{A}^2\Tr{B^2}\geqslant0,\Tr{A^2}\Tr{B^2}\geqslant0$, i.e.,
$\Omega_j(A,B)\geqslant0$ for $j=1,2,4,5$ by the definition. In
addition, $\Tr{A^2B^2}=\Tr{BA^2B}\geqslant 0$ since
$BA^2B\geqslant0$. Thus $\Omega_7(A,B)\geqslant0$. Consider the
operator $X=AB+BA$. Clearly $X$ is a Hermitian operator. Moreover
$X^2\geqslant0$, thus $\Tr{X^2}\geqslant0$. Because $\Tr{X^2}
=2(\Tr{ABAB} - \Tr{A^2B^2})$, we have that
$\Tr{ABAB}\geqslant\Tr{A^2B^2}\geqslant0$. Hence
$\Omega_8(A,B)\geqslant0$. In summary, $\Omega_j(A,B)\geqslant0$ for
$j=1,2,4,5,7,8$. However, $\Omega_3(A,B)$ and $\Omega_6(A,B)$ are
not always non-negative.
\end{remark}

\begin{remark}\label{rem:free}
The rhs of \eqref{eq:8-symmetric} remind us of one of applications
to random matrix theory from free probability theory, established by
Voiculescu \cite{voiculescu2000}. Specifically, we can consider two
independent random observables $A$ and $B$ from Gaussian unitary
ensemble (GUE), according to free probability theory, $A$ and $B$
are asymptotic free (see the meaning of freeness in
\cite{nica2006}). Indeed, denote $\varphi(\cdot)=\frac1d\Tr{\cdot}$,
where $\Tr{\cdot}$ means the trace of matrix, when $d$ becomes large
enough, we have
\begin{eqnarray}
\varphi(ABAB)\simeq \varphi(A^2)\varphi(B)^2 +
\varphi(A)^2\varphi(B^2) - \varphi(A)^2\varphi(B)^2,
\end{eqnarray}
that is,
\begin{eqnarray}
\Omega_8(A,B)\simeq d^{-2}\Omega_2(A,B) - d^{-3}\Omega_1(A,B).
\end{eqnarray}
Similarly, we have
\begin{eqnarray}
&&\Omega_3(A,B)\simeq d^{-1}\Omega_1(A,B),\\
&&\Omega_5(A,B)\simeq d^{-2}\Omega_1(A,B),\\
&&\Omega_6(A,B)\simeq d^{-1}\Omega_2(A,B),\\
&&\Omega_7(A,B)\simeq d^{-1}\Omega_4(A,B).
\end{eqnarray}
Furthermore, we obtain that
\begin{eqnarray}
\int\Delta A(U\Lambda U^\dagger)^2\cdot\Delta B(U\Lambda
U^\dagger)^2\dif\mu_{\mathrm{Haar}}(U) &\simeq &
\Pa{\omega_1(\Lambda)+d^{-1}\omega_3(\Lambda)+d^{-2}\omega_5(\Lambda)-d^{-3}\omega_8(\Lambda)}\Omega_1(A,B)\notag\\
&&+\Pa{\omega_2(\Lambda)+d^{-1}\omega_6(\Lambda)+d^{-2}\omega_8(\Lambda)}\Omega_2(A,B)\notag\\
&&+\Pa{\omega_4(\Lambda)+d^{-1}\omega_7(\Lambda)}\Omega_4(A,B).
\end{eqnarray}
The calculation in Theorem~\ref{th:iso}, and the subsequent remark
suggest us that there are three terms, i.e., $\Omega_1(A,B),
\Omega_2(A,B)$, and $\Omega_4(A,B)$, as the dimension grows large,
play a leading role in estimating the average of uncertainty-product
within isospectral density matrices. This also tells us that if we
want to get a better lower bound about uncertainty-product, then
when we take average of any improved lower bound, we should get
larger coefficients of such three terms.

Besides, for a fixed $\Lambda$, we may view the left hand side of
\eqref{eq:8-symmetric} as a function of two random observables $A$
and $B$, for instance, from GUE or Wishart ensemble. We can also
consider the concentration of measure phenomenon about such two
observables. We leave these questions in the future research.
\end{remark}

\subsection{Average of uncertainty-product on pure states}

For the pure state case, the average of uncertainty-product is
easier to calculate. What we have obtained is the following:

\begin{thrm}\label{th:pure-case}
For two observables $A$ and $B$ on $\cH_d$, the average of
uncertainty-product taken over the whole set of all pure states in
$\cH_d$ is given by
\begin{eqnarray}
\int\Delta A(\psi)^2\cdot\Delta B(\psi)^2\dif\mu(\psi) =
\sum^8_{j=1}u_j\Omega_j(A,B),
\end{eqnarray}
where $\Omega_j(A,B)$ is from Theorem~\ref{th:iso}, and for
$K_d=(d(d+1)(d+2)(d+3))^{-1}$,
\begin{eqnarray}
&&u_1 = K_d,\quad u_2=-(d+2)K_d,\quad u_3 =4K_d,\quad u_4=(d^2+3d+1)K_d,\\
&&u_5 = 2K_d, \quad u_6=-2(d+1)K_d,\quad u_7=(d^2+d-2)K_d, \quad
u_8=2K_d.
\end{eqnarray}
We also have that
\begin{eqnarray}
\int\dif\mu(\psi)\Br{\Pa{\langle\set{A,B}\rangle_\psi-\langle
A\rangle_\psi\langle B\rangle_\psi}^2 + \langle [A,B]\rangle^2_\psi}
= \sum^8_{j=1}l_j\Omega_j(A,B),
\end{eqnarray}
where
\begin{eqnarray}
&&l_1 = K_d,\quad l_2 = K_d,\quad l_3 = -2(d+1)K_d,\quad l_4 = K_d,\quad l_5 = (d+1)(d+2)K_d, \\
&&\quad l_6 = -2(d+1)K_d,\quad l_7 = (d^2+d-2)K_d, \quad l_8 =
-2(2d+5)K_d.
\end{eqnarray}
\end{thrm}

In the above theorem, we investigate average behavior of both sides
of Heisenberg-Robertson-Kennard relations on uniform pure state
ensemble. For the case of the average of product of uncertainties
(or the corresponding lower bounds for this quantity) over pure
Haar-distributed quantum states, the corresponding integrals are
very easy to perform as the integral
$$
\int\psi^{\ot k}\dif\mu(\psi)
$$
involved in all the averages are proportional to the projectors on
the symmetric powers of the relevant Hilbert space. This will be
clear in the proof, see \eqref{eq:symmetric-proj}. The details of
the proof of Theorem~\ref{th:pure-case} can be found in
Subsection~\ref{subsect:pure-case}.

\begin{remark}
In higher dimensional space, there are two terms playing major role
in the average uncertainty-product relative to other terms, i.e.,
$\Omega_4(A,B)=\Tr{A^2}\Tr{B^2}$ and $\Omega_7(A,B)=\Tr{A^2B^2}$.
However, the terms which play major role in the average lower bound
of uncertainty-product is $\Omega_5(A,B)=\Tr{AB}\Tr{AB}$, and
$\Omega_7(A,B)=\Tr{A^2B^2}$. Furthermore, we can derive that
\begin{eqnarray}\label{eq:non-negative}
&&\int\dif\mu(\psi)\Set{\Delta A(\psi)^2\cdot\Delta B(\psi)^2 -
\Br{\Pa{\langle\set{A,B}\rangle_\psi-\langle A\rangle_\psi\langle
B\rangle_\psi}^2 + \langle [A,B]\rangle^2_\psi}}\notag\\
&&=
-(d+3)K_d\Omega_2(A,B)+2(d+3)K_d\Omega_3(A,B)+d(d+3)K_d\Omega_4(A,B)\notag\\
&&~~~-d(d+3)K_d\Omega_5(A,B)+4(d+3)K_d\Omega_8(A,B).
\end{eqnarray}
By the nonnegativity of the left hand side of
\eqref{eq:non-negative}, we get the following inequality:
\begin{eqnarray}\label{eq:positive-diference-pure}
2\Omega_3(A,B)+d\Omega_4(A,B)+4\Omega_8(A,B)\geqslant\Omega_2(A,B)+d\Omega_5(A,B).
\end{eqnarray}
That is,
\begin{eqnarray}\label{eq:trace-ineq}
&&2\Tr{AB}\Tr{A}\Tr{B}+d\Tr{A^2}\Tr{B^2}+4\Tr{ABAB}\notag\\
&&\geqslant\Tr{A^2}\Tr{B}^2
+\Tr{A}^2\Tr{B^2}+d\Tr{AB}\Tr{AB}.
\end{eqnarray}
It seems difficult to show the above matrix trace inequality
\eqref{eq:trace-ineq} directly. This inequality about two
observables is what we want to get, i.e., uncertainty relation which
is independent of state.
\end{remark}

\begin{remark}
Naturally, a pure state $\ket{\psi}$ is called the \emph{average
state} with respect to uncertainty product of observables $(A,B)$ if
it satisfies that
\begin{eqnarray}\label{eq:pure-state-of-ave}
\Delta A(\psi)^2\cdot\Delta B(\psi)^2 =
\sum^8_{j=1}u_j\Omega_j(A,B).
\end{eqnarray}
What properties do such state have? Answering this question can
reveal principally why we do not need to take any measurements, and
we can guess the uncertainty about observables by taking average.
\end{remark}

\begin{cor}
For two observables $A$ and $B$ on $\complex^2$, the average of
uncertainty-product taken over the whole set of all pure states is
given by
\begin{eqnarray}
\int\Delta A(\psi)^2\cdot\Delta B(\psi)^2\dif\mu(\psi) &=&
\frac1{120}\Omega_1(A,B) - \frac1{30}\Omega_2(A,B) +
\frac1{30}\Omega_3(A,B) + \frac{11}{120}\Omega_4(A,B)\notag\\
&&+\frac1{60}\Omega_5(A,B) -\frac1{20}\Omega_6(A,B) +
\frac1{30}\Omega_7(A,B) + \frac1{60}\Omega_8(A,B).
\end{eqnarray}
We also have that
\begin{eqnarray}
&&\int\dif\mu(\psi)\Br{\Pa{\langle\set{A,B}\rangle_\psi-\langle
A\rangle_\psi\langle B\rangle_\psi}^2 + \langle
[A,B]\rangle^2_\psi}\notag\\
&&=\frac1{120}\Omega_1(A,B) + \frac1{120}\Omega_2(A,B) -
\frac1{20}\Omega_3(A,B) + \frac1{120}\Omega_4(A,B)\notag\\
&&~~~~ + \frac1{10}\Omega_5(A,B)- \frac1{20}\Omega_6(A,B) +
\frac1{30}\Omega_7(A,B) - \frac3{20}\Omega_8(A,B).
\end{eqnarray}
\end{cor}

Next, as an example, we take $A=\sigma_i$ and $B=\sigma_j$, where
$\sigma_i$ and $\sigma_j$ are any two different matrices from three
Pauli's matrices, using the above Corollary, then we get the average
of uncertainty-product of $A$ and $B$ is given by
\begin{eqnarray}
\int\Delta A(\psi)^2\cdot\Delta B(\psi)^2\dif\mu(\psi) = \frac25.
\end{eqnarray}
Moreover,
\begin{eqnarray}
\int\dif\mu(\psi)\Br{\Pa{\langle\set{A,B}\rangle_\psi-\langle
A\rangle_\psi\langle B\rangle_\psi}^2 + \langle
[A,B]\rangle^2_\psi}=\frac25.
\end{eqnarray}
This is surprising! As we have seen that the following inequality
\begin{eqnarray}
\Delta A(\psi)^2\cdot\Delta B(\psi)^2\geqslant
\Pa{\langle\set{A,B}\rangle_\psi-\langle A\rangle_\psi\langle
B\rangle_\psi}^2 + \langle [A,B]\rangle^2_\psi
\end{eqnarray}
holds for all pure state $\ket{\psi}$. From the above discussion, we
see that
\begin{eqnarray}
\int\dif\mu(\psi) f(\psi) =0,
\end{eqnarray}
where $f$ is defined by
\begin{eqnarray}
f(\psi) = \Delta A(\psi)^2\cdot\Delta B(\psi)^2-\Br{
\Pa{\langle\set{A,B}\rangle_\psi-\langle A\rangle_\psi\langle
B\rangle_\psi}^2 + \langle [A,B]\rangle^2_\psi},
\end{eqnarray}
which is obviously a non-negative function of the pure state
$\ket{\psi}$. By Lebesgue integration theory, we get that $f(\psi)$
vanishes almost everywhere except a zero-measure subset of all pure
states. In other words,
\begin{eqnarray}\label{eq:AAA}
\Delta A(\psi)^2\cdot\Delta B(\psi)^2 =
\Pa{\langle\set{A,B}\rangle_\psi-\langle A\rangle_\psi\langle
B\rangle_\psi}^2 + \langle [A,B]\rangle^2_\psi,\quad\text{a.e.}
\end{eqnarray}
From the above observation, we see that any desire to improve
universally the uncertainty-product seems impossible, at least in
the qubit case for two observables $\sigma_i$ and $\sigma_j$ chosen
from three Pauli's matrices.

\subsection{Average of uncertainty-product on the mixed states}

For the mixed state, comparing with the pure state, the calculation
is more complicated, we have the following result.
\begin{thrm}\label{th:HS-average}
For two observables $A$ and $B$ on $\cH_d$, the average of
uncertainty-product taken over the whole set of all density matrices
$\density{\cH_d}$ is given by
\begin{eqnarray}\label{eq:mixed-ave}
&&\int\Delta A(\rho)^2\cdot\Delta
B(\rho)^2\dif\mu_{\rH\rS}(\rho)=\sum^8_{j=1}\overline{\omega}_j
\cdot\Omega_j(A,B),
\end{eqnarray}
where $\overline{\omega_j}=\int\omega_j(\Lambda)\dif\nu(\Lambda)
(j=1,\ldots,8)$.
\end{thrm}

\begin{proof}
The proof follows directly from Theorem~\ref{th:iso} by using
Proposition~\ref{prop:karol} and Lemma~\ref{lem:lin}.
\end{proof}

\begin{remark}
In fact, we can give the final formulae for $\overline{\omega_j}$'s.
We ignore the tedious but simple calculations.
\begin{eqnarray}
\overline{\omega}_1 &=&
N_d\Pa{d^4-20d^2+158-\frac{50}{d^2+1}+\frac{792}{d^2+2}-\frac{1512}{d^2+3}},\\
\overline{\omega}_2 &=&
N_d\Pa{-d^5+18d^3-118d-\frac{50d}{d^2+1}+\frac{504d}{d^2+3}},\\
\overline{\omega}_3 &=&
N_d\Pa{4d^3-80d+\frac{200d}{d^2+1}-\frac{1584d}{d^2+2}+\frac{2016d}{d^2+3}},\\
\overline{\omega}_4 &=&
N_d\Pa{d^6-17d^4+99d^2-316-\frac{50}{d^2+1}+\frac{396}{d^2+2}+\frac{504}{d^2+3}},\\
\overline{\omega}_5 &=& N_d\Pa{2d^2-40+\frac{100}{d^2+1}-\frac{792}{d^2+2}+\frac{1008}{d^2+3}},\\
\overline{\omega}_6 &=& N_d\Pa{-2d^4+38d^2-276-\frac{792}{d^2+2}+\frac{2016}{d^2+3}},\\
\overline{\omega}_7 &=& N_d\Pa{2d^5-d^4-28d^3+9d^2+136d+\frac{100d}{d^2+1}-\frac{672d}{d^2+3}},\\
\overline{\omega}_8 &=&
N_d\Pa{2d-\frac{100d}{d^2+1}+\frac{396d}{d^2+2}-\frac{336d}{d^2+3}}.
\end{eqnarray}
From the above formulae, we can see that in higher dimensional
space, $\Omega_4(A,B)=\Tr{A^2}\Tr{B^2}$ plays a leading role
relative to other terms. We also see from Remark~\ref{rem:free}
that, for the large enough dimension $d$, when observables $A$ and
$B$ taken from GUE are independent,
\begin{eqnarray}
\int\Delta A(\rho)^2\cdot\Delta
B(\rho)^2\dif\mu_{\rH\rS}(\rho)\simeq
m_1\Omega_1(A,B)+m_2\Omega_2(A,B)+m_4\Omega_4(A,B).
\end{eqnarray}
where
\begin{eqnarray}
m_1 &=& \overline{\omega}_1+d^{-1}\overline{\omega}_3+d^{-2}\overline{\omega}_5-d^{-3}\overline{\omega}_8,\\\
m_2 &=& \overline{\omega}_2+d^{-1}\overline{\omega}_6+d^{-2}\overline{\omega}_8,\\
m_4 &=& \overline{\omega}_4+d^{-1}\overline{\omega}_7.
\end{eqnarray}
Similar to the pure state case (see \eqref{eq:pure-state-of-ave}), a
mixed state $\rho$ is called the \emph{average state} with respect
to uncertainty product of observables $(A,B)$ if it satisfies that
\begin{eqnarray}
\Delta A(\rho)^2\cdot\Delta
B(\rho)^2=\sum^8_{j=1}\overline{\omega}_j \cdot\Omega_j(A,B).
\end{eqnarray}
We can ask analogous problems parallel to the pure state case. But
we are not concerned these problems in this paper.
\end{remark}

\subsection{Average lower bound of uncertainty-product}

Here we also calculate the average of the lower bound of
uncertainty-product in \eqref{eq:R-S}.
\begin{thrm}\label{th:ave-lower}
For two observables $A$ and $B$ on $\cH_d$,  it holds that
\begin{eqnarray}\label{4}
\int_{\density{\cH_d}}\dif\mu_{\rH\rS}(\rho)\Pa{\langle\set{A,B}\rangle_\rho-\langle
A\rangle_\rho\langle B\rangle_\rho}^2 =
\sum^8_{j=1}\beta_j\Omega_j(A,B),
\end{eqnarray}
where $N_d^{-1}=d^2(d^2-1)(d^2-4)(d^2-9)$ and
\begin{eqnarray}
\beta_1 &=& N_d\Pa{d^4-18d^2+158-\frac{50}{d^2+1}+\frac{792}{d^2+2}-\frac{1512}{d^2+3}},\\
\beta_2 &=& N_d\Pa{d^3-20d+\frac{50d}{d^2+1}-\frac{396d}{d^2+2}+\frac{504d}{d^2+3}},\\
\beta_3 &=&
N_d\Pa{-2d^5+38d^3-276d-\frac{792d}{d^2+2}+\frac{2016d}{d^2+3}},\\
\beta_4 &=& N_d\Pa{-2d^2-20+\frac{50}{d^2+1}-\frac{396}{d^2+2}+\frac{504}{d^2+3}},\\
\beta_5 &=&
N_d\Pa{d^6-15d^4-2d^3+60d^2+34d-140+\frac{200d+200}{d^2+1}-\frac{396d+792}{d^2+2}+\frac{1008}{d^2+3}},\\
\beta_6 &=&
N_d\Pa{-2d^3+4d^2+34d-380+\frac{200d+500}{d^2+1}-\frac{396d+1584}{d^2+2}+\frac{2016}{d^2+3}},\\
\beta_7 &=&
N_d\Pa{-d^7+\frac{27}{2}d^5-\frac{91}{2}d^3+70d-\frac{50d}{d^2+1}+\frac{396d}{d^2+2}-\frac{672d}{d^2+3}},\\
\beta_8 &=&
N_d\Pa{-d^7+\frac{27}{2}d^5-\frac{91}{2}d^3+68d+\frac{50d}{d^2+1}-\frac{336d}{d^2+3}}.
\end{eqnarray}
Thus
\begin{eqnarray}\label{5}
\int_{\density{\cH_d}}\dif\mu_{\rH\rS}(\rho)\Br{\Pa{\langle\set{A,B}\rangle_\rho-\langle
A\rangle_\rho\langle B\rangle_\rho}^2 +\langle [A,B]\rangle^2_\rho}=
\sum^8_{j=1}\beta'_j\Omega_j(A,B),
\end{eqnarray}
where
\begin{eqnarray}
&&\beta'_1=\beta_1,~ \beta'_2 =\beta_2,~ \beta'_3 = \beta_3,~
\beta'_4 = \beta_4,~ \beta'_5 = \beta_5,~ \beta'_6 = \beta_6,
 \\ &&\beta'_7 =
N_d \Pa{-d^7+14d^5-53d^3+102d-\frac{100d}{d^2+1}+\frac{396d}{d^2+2}-\frac{672d}{d^2+3}},\\
&&\beta'_8 =
N_d \Pa{-d^7+13d^5-38d^3+36d+\frac{100d}{d^2+1}-\frac{336d}{d^2+3}}.
\end{eqnarray}
\end{thrm}

The average of the lower bound of uncertainty-product can be the
reference value for improving the lower bound of
uncertainty-product, as suggested in Section~\ref{sect:motivation}.
The proof of Theorem~\ref{th:ave-lower} is put in
Subsection~\ref{subset:ave-lower}.

\begin{remark}
From the above Theorem~\ref{th:ave-lower}, we see that in higher
dimensional space, $\Omega_7(A,B)=\Tr{A^2B^2}$ and
$\Omega_8(A,B)=\Tr{ABAB}$ play a leading role relative to other
terms.
\end{remark}

\begin{remark}
We can still compare \eqref{eq:mixed-ave} and \eqref{5} in order to
obtain another matrix trace inequality:
\begin{eqnarray}\label{eq:positive-diference-mixed}
\sum^8_{j=1}(\overline{\omega}_j - \beta'_j)
\cdot\Omega_j(A,B)\geqslant 0.
\end{eqnarray}
As a matter of fact, \eqref{eq:positive-diference-pure} and
\eqref{eq:positive-diference-mixed} are just two special cases of
the following matrix trace inequalities:
\begin{eqnarray}
\sum^8_{j=1}f_j(d) \cdot\Omega_j(A,B)\geqslant 0,
\end{eqnarray}
where $f_j(d)(j=1,\ldots,8)$ are the dimension-dependent factors
under some constraints.
\end{remark}

\section{Concentration of measure phenomenon}\label{sect:concentration}

In order to discuss the concentration of measure phenomenon might
being happened to the uncertainty-product, we will use the
concentration of measure phenomenon on the special unitary group
$\rS\rU(\cH_d)$, established recently by Oszmaniec in his thesis
\cite{oszmaniec2014phd}.
\begin{prop}[Concentration of measure on $\rS\rU(\cH_d)$]
Consider a special unitary group $\rS\rU(\cH_d)$ equipped with the
Haar measure $\mu_{\mathrm{Haar}}$ and a Riemann metric
$g_{\rH\rS}$. Let $f:\rS\rU(\cH_d)\to\real$ be a smooth function on
$\rS\rU(\cH_d)$ with the mean $\bar
f=\int_{\rS\rU(\cH_d)}f(U)\dif\mu_{\mathrm{Haar}}(U)$, let
\begin{eqnarray}
L = \sqrt{\max\Set{g_{\rH\rS}(\nabla f,\nabla f):
U\in\rS\rU(\cH_d)}}
\end{eqnarray}
be the Lipschitz constant of $f$. Then, for every
$\epsilon\geqslant0$, the following concentration inequalities hold
\begin{eqnarray}
\mu_{\mathrm{Haar}}\Set{U\in\rS\rU(\cH_d):f(U) - \bar f\geqslant
\epsilon}&\leqslant& \exp\Pa{-\frac{d\epsilon^2}{4L^2}},\\
\mu_{\mathrm{Haar}}\Set{U\in\rS\rU(\cH_d):f(U) - \bar f\leqslant
-\epsilon}&\leqslant& \exp\Pa{-\frac{d\epsilon^2}{4L^2}}.
\end{eqnarray}
\end{prop}

Denote
\begin{eqnarray}
\Phi(U)= \Delta A(U\rho U^\dagger)^2\cdot \Delta B(U\rho
U^\dagger)^2.
\end{eqnarray}
From \eqref{eq:prod-uncertainty}, we see that
\begin{eqnarray}
\Phi(U) = \Tr{U^{\ot4}\rho^{\ot 4}U^{\dagger,\ot4}\Br{A^2\ot
B^2\ot\I^{\ot2}+A^{\ot 2}\ot B^{\ot 2} - A^2\ot B^{\ot 2}\ot\I-
A^{\ot 2}\ot B^2\ot\I}}.
\end{eqnarray}
By using the result in \cite[Lemma~6.1]{oszmaniec2014phd}, we see
that the Lipschitz constant $L_\Phi$ of the function $\Phi$, with
respect to the metric tensor $g_{\rH\rS}$, satisfies
\begin{eqnarray}
L_\Phi&\leqslant& 8\norm{A^2\ot B^2\ot\I^{\ot2}+A^{\ot 2}\ot B^{\ot
2}
- A^2\ot B^{\ot 2}\ot\I- A^{\ot 2}\ot B^2\ot\I}_\infty\notag\\
&\leqslant& 32\norm{A}^2_\infty\norm{B}^2_\infty.
\end{eqnarray}
Thus we have the following result:
\begin{thrm}[Concentration of measure within isospectral density matrices]
For every $\epsilon\geqslant0$, the following concentration
inequalities hold
\begin{eqnarray}
\mu_{\mathrm{Haar}}\Set{U\in\rS\rU(\cH_d):\Phi(U) -
\overline{\Phi}\geqslant
\epsilon}&\leqslant& \exp\Pa{-\frac{d\epsilon^2}{4096\norm{A}^4_\infty\norm{B}^4_\infty}},\\
\mu_{\mathrm{Haar}}\Set{U\in\rS\rU(\cH_d):\Phi(U) -
\overline{\Phi}\leqslant -\epsilon}&\leqslant&
\exp\Pa{-\frac{d\epsilon^2}{4096\norm{A}^4_\infty\norm{B}^4_\infty}}.
\end{eqnarray}
\end{thrm}
This result shows that when we consider the uncertainty-product for
two bounded observables $A$ and $B$ over the set of isospectral
density matrices, the uncertainty-product around its average, in
\eqref{eq:8-symmetric}
\begin{eqnarray}
\overline{\Phi} = \sum^8_{j=1}\omega_j(\Lambda)\cdot\Omega_j(A,B),
\end{eqnarray}
has an overwhelming probability.

\begin{lem}[L\'{e}vy's lemma]
Let $f:\mathbb{S}^k\to\real$ be a Lipschitz function from $k$-sphere
to real line with the Lipschitz constant $L$ (with respect to the
Euclidean norm) and a point $u\in\mathbb{S}^k$ be chosen uniformly
at random. Then, for all $\epsilon>0$,
\begin{eqnarray}
\mathbf{Pr}\Set{\abs{f(u)-\bar f}>\epsilon}\leqslant
2\exp\Pa{-\frac{(k+1)\epsilon^2}{9\pi^3L^2\ln2}},
\end{eqnarray}
where $\bar f:=\int_{\mathbb{S}^k}f(u)\dif \mu(u)$ means the mean
value of $f$ with respect to uniform probability measure on the unit
sphere $\mathbb{S}^k$.
\end{lem}

Let $f(\rho) = \Delta A(\rho)^2\cdot \Delta B(\rho)^2$. Then
\begin{eqnarray}
f(\rho) - f(\sigma) &=& \Tr{\Br{A^2\ot B^2}\Br{\rho^{\ot
2}-\sigma^{\ot 2}}} + \Tr{\Br{A^{\ot 2}\ot B^{\ot 2}}\Br{\rho^{\ot
4}-\sigma^{\ot 4}}} \notag\\
&&- \Tr{\Br{A^2\ot B^{\ot 2}}\Br{\rho^{\ot 3}-\sigma^{\ot 3}}} -
\Tr{\Br{B^2\ot A^{\ot 2}}\Br{\rho^{\ot 3}-\sigma^{\ot 3}}}.
\end{eqnarray}
Thus
\begin{eqnarray}
\abs{f(\rho) - f(\sigma)} &\leqslant& \abs{\Tr{\Br{A^2\ot
B^2}\Br{\rho^{\ot 2}-\sigma^{\ot 2}}}} + \abs{\Tr{\Br{A^{\ot 2}\ot
B^{\ot 2}}\Br{\rho^{\ot
4}-\sigma^{\ot 4}}}} \notag\\
&&+ \abs{\Tr{\Br{A^2\ot B^{\ot 2}}\Br{\rho^{\ot 3}-\sigma^{\ot 3}}}}
+ \abs{\Tr{\Br{B^2\ot A^{\ot 2}}\Br{\rho^{\ot 3}-\sigma^{\ot 3}}}}\notag\\
&\leqslant& \norm{A^2\ot B^2}_\infty\norm{\rho^{\ot 2}-\sigma^{\ot
2}}_1 + \norm{A^{\ot 2}\ot B^{\ot 2}}_\infty\norm{\rho^{\ot
4}-\sigma^{\ot 4}}_1\notag\\
&&+\norm{A^2\ot B^{\ot 2}}_\infty\norm{\rho^{\ot 3}-\sigma^{\ot
3}}_1+\norm{B^2\ot A^{\ot 2}}_\infty\norm{\rho^{\ot 3}-\sigma^{\ot
3}}_1.
\end{eqnarray}
Since
\begin{eqnarray}
\norm{\rho^{\ot k} - \sigma^{\ot k}}_1\leqslant
k\norm{\rho-\sigma}_1,
\end{eqnarray}
it follows that
\begin{eqnarray}
\abs{f(\rho) - f(\sigma)} \leqslant
\Pa{12\norm{A}^2_\infty\norm{B}^2_\infty}\norm{\rho-\sigma}_1
\end{eqnarray}

For the pure states, that is, $\rho=\out{\psi}{\psi}$ and
$\sigma=\out{\phi}{\phi}$, we have $\norm{\psi-\phi}_1 \leqslant
\sqrt{2}\norm{\psi-\phi}_2$, implying
\begin{eqnarray}
\abs{f(\psi) - f(\phi)} \leqslant L\cdot\norm{\psi-\phi}_2,
\end{eqnarray}
where $L:=12\sqrt{2}\norm{A}^2_\infty\norm{B}^2_\infty$. Note here
that $k=2d-1$ since pure states live in $\complex^d$. Then
\begin{eqnarray}
\mathbf{Pr}\Set{\abs{f(\psi)-\bar f}>\epsilon}\leqslant
2\exp\Pa{-\frac{d\epsilon^2}{1296\pi^3\norm{A}^4_\infty\norm{B}^4_\infty\ln2}}.
\end{eqnarray}
When $\norm{A}_\infty$ and $\norm{B}_\infty$ are independent of the
dimension $d$, it shows the concentration of measure phenomenon.

In fact, we can view $\rho$ and $\sigma$ in $\density{\cH_d}$ as
reduced states of Haar-distributed bipartite states
$\ket{\Psi_\rho}$ and $\ket{\Psi_\sigma}$ in $\cH_d\ot\cH_d$, then
let $g(\Psi_\rho)=\Delta A(\rho)^2\cdot \Delta B(\rho)^2$, where
$\Psi_\rho=\out{\Psi_\rho}{\Psi_\rho}$ and
$\rho=\Ptr{2}{\out{\Psi_\rho}{\Psi_\rho}}$. Thus
\begin{eqnarray}
\abs{g(\Psi_\rho) - g(\Psi_\sigma)} \leqslant
(12\sqrt{2}\norm{A}^2_\infty\norm{B}^2_\infty)\norm{\Psi_\rho -
\Psi_\sigma}_2.
\end{eqnarray}
Then
\begin{eqnarray}
\mathbf{Pr}\Set{\abs{f(\rho)-\bar f}>\epsilon}\leqslant
2\exp\Pa{-\frac{d^2\epsilon^2}{1296\pi^3\norm{A}^4_\infty\norm{B}^4_\infty\ln2}}.
\end{eqnarray}
Thus we have the following:
\begin{thrm}[Concentration of measure]
Assume that $\norm{A}_\infty$ and $\norm{B}_\infty$ are independent
of dimension, where $A$ and $B$ are bounded observables. It holds
that
\begin{eqnarray}
\mathbf{Pr}\Set{\abs{\Delta A(\psi)^2\cdot \Delta B(\psi)^2 -
\langle \Delta A(\psi)^2\cdot \Delta B(\psi)^2
\rangle}>\epsilon}\leqslant
2\exp\Pa{-\frac{d\epsilon^2}{1296\pi^3\norm{A}^4_\infty\norm{B}^4_\infty\ln2}}
\end{eqnarray}
and
\begin{eqnarray}
\mathbf{Pr}\Set{\abs{\Delta A(\rho)^2\cdot \Delta B(\rho)^2 -
\langle \Delta A(\rho)^2\cdot \Delta B(\rho)^2
\rangle}>\epsilon}\leqslant
2\exp\Pa{-\frac{d^2\epsilon^2}{1296\pi^3\norm{A}^4_\infty\norm{B}^4_\infty\ln2}}.
\end{eqnarray}
Here $\langle f(\rho)\rangle =\int f(\rho) \dif\mu_{\rH\rS}(\rho)$.
\end{thrm}
Generally, observables $A$ and $B$ are dimension-dependent, thus we
cannot obtain the concentration of measure phenomenon universally.
But of course, even though $A$ and $B$ are dimension-dependent, we
could still get the concentration of measure phenomenon, for
instance, whenever their operator norms are uniformly bounded.
Besides, inequalities presented above do not have to be tight, i.e.,
even if the right hand side is "large", the relevant left hand side
might still be very small.

\section{Concluding remarks}

This paper deals with uncertainty relations in various random state
ensembles. As suggested, taking a state at random also corresponds
to assuming minimal prior knowledge about the system in question. We
make an attempt in describing uncertainty relation using only
observables by taking average of uncertainty-product of any two
bounded observables in our random state ensemble (see
\eqref{eq:positive-diference-pure}). We also establish the
typicality of a random state with respect to any two bounded
observables under restricted conditions. The concentration of
measure phenomenon is a very important property for a random state
since it predicates the bulk behavior of a large number of quantum
particles without any practical detections. Theoretically, sampled
states randomly will show up average behavior with respect to a pair
of bounded observables as we increases the level of the quantum
system under consideration. In addition, we have also present an
interesting result: beyond the set of zero-measure of all pure qubit
states, it holds that
\begin{eqnarray}
\Delta A(\psi)^2\cdot\Delta B(\psi)^2 =
\Pa{\langle\set{A,B}\rangle_\psi-\langle A\rangle_\psi\langle
B\rangle_\psi}^2 + \langle [A,B]\rangle^2_\psi.
\end{eqnarray}
This result indicates that any desire to improve the
uncertainty-product universally seems impossible, at least in the
qubit case for two distinct observables $\sigma_i$ and $\sigma_j$
chosen from three Pauli's matrices. Our calculations can help us
check how large the gap is between the uncertainty-product and any
obtained lower bounds about the uncertainty-product. We hope the
results obtained in this paper will shed new light on quantum
information processing tasks.

\subsubsection*{Acknowledgements}
L.Z. is supported Natural Science Foundation of Zhejiang Province of
China (LY17A010027), and by National Natural Science Foundation of
China (No.11301124), and also supported by the cross-disciplinary
innovation team building project of Hangzhou Dianzi University.
J.M.W. is also supported by NSFC (No.11401007). L.Z. also would like
to thank Shao-Ming Fei and Naihuan Jing for providing some important
remarks on this paper, and thank both Yichen Huang and S. Salimi for
their comments on the first version of our manuscript. Both authors
would like to thank the Referee for reading our manuscript very
carefully and for insightful comments in improving the presentation
of the results.


\section{Appendix: the computation of $\sE_k(\Lambda)$}\label{sect:appendix}

Consider a system of $k$ qudits, each with a standard local
computational basis $\set{\ket{i},i=1,\ldots,d}$. The Schur-Weyl
duality relates transforms on the system performed by local
$d$-dimensional unitary operations to those performed by permutation
of the qudits. Recall that the symmetric group $S_k$ is the group of
all permutations of $k$ objects. This group is naturally represented
in our system by
\begin{eqnarray}
\bP(\pi)\ket{i_1\cdots i_k} := \ket{i_{\pi^{-1}(1)}\cdots
i_{\pi^{-1}(k)}},
\end{eqnarray}
where $\pi\in S_k$ is a permutation and $\ket{i_1\cdots i_k}$ is
shorthand for $\ket{i_1}\ot\cdots\ot\ket{i_k}$. Let $\unitary{d}$ be
the group of $d\times d$ unitary operators. This group is naturally
represented in our system by
\begin{eqnarray}
\bQ(U)\ket{i_1\cdots i_k} := U\ket{i_1}\ot\cdots\ot U\ket{i_k},
\end{eqnarray}
where $U\in\unitary{d}$. Thus we have the following famous result:
\begin{thrm}[Schur]\label{th-schur}
Let $\cA = \spn\Set{\bP(\pi): \pi\in S_k}$ and $\cB =
\spn\Set{\bQ(U): U\in\unitary{d}}$. Then:
\begin{eqnarray}
\cA' = \cB\quad\text{and}\quad \cA = \cB'
\end{eqnarray}
\end{thrm}

The following result concerns with a wonderful decomposition of the
representations on $k$-fold tensor space $(\complex^d)^{\ot k}$ of
$\unitary{d}$ and $S_k$, respectively, using their corresponding
irreps accordingly.

\begin{thrm}[Schur-Weyl duality]\label{th:S-W-Duality}
There exist a basis, known as Schur basis, in which representation
$\Pa{\bQ\bP,(\complex^d)^{\ot k}}$ of $\unitary{d}\times S_k$
decomposes into irreducible representations $\bQ_\lambda$ and
$\bP_\lambda$ of $\unitary{d}$ and $S_k$, respectively:
\begin{enumerate}[(i)]
\item $(\complex^d)^{\ot k}\cong \bigoplus_{\lambda\vdash(k,d)} \bQ_\lambda\ot
\bP_\lambda$;
\item $\bP(\pi)\cong \bigoplus_{\lambda\vdash(k,d)} \I_{\bQ_\lambda}\ot
\bP_\lambda(\pi)$;
\item $\bQ(U)\cong \bigoplus_{\lambda\vdash(k,d)} \bQ_\lambda(U)\ot
\I_{\bP_\lambda}$.
\end{enumerate}
Since $\bQ$ and $\bP$ commute, we can define representation
$\Pa{\bQ\bP,(\complex^d)^{\ot k}}$ of $\unitary{d}\times S_k$ as
\begin{eqnarray}
\bQ\bP(U,\pi) = \bQ(U)\bP(\pi) = \bP(\pi)\bQ(U)\quad \forall
(U,\pi)\in \unitary{d}\times S_k.
\end{eqnarray}
Then:
\begin{eqnarray}\label{schur-duality}
\bQ\bP(U,\pi) = U^{\ot k}P_\pi = P_\pi U^{\ot k}\cong
\bigoplus_{\lambda\vdash(k,d)} \bQ_\lambda(U)\ot \bP_\lambda(\pi).
\end{eqnarray}
\end{thrm}

The dimensions of pairing irreps for $\unitary{d}$ and $S_k$,
respectively, in Schur-Weyl duality can be computed by so-called
\emph{hook length formulae}. The hook of box $(i, j)$ in a Young
diagram determined by a partition $\lambda$ is given by the box
itself, the boxes to its right and below. The hook length is the
number of boxes in a hook. Specifically, we have the following
result without its proof:
\begin{thrm}[Hook length formulae]
The dimensions of pairing irreps for $\unitary{d}$ and $S_k$,
respectively, in Schur-Weyl duality can be given as follows:
\begin{eqnarray}
\dim(\bQ_\lambda) &=& \prod_{(i,j)\in\lambda}\frac{d+j-i}{h(i,j)} =
\prod_{1\leqslant i<j\leqslant
d}\frac{\lambda_i - \lambda_j + j-i}{j-i}, \\
\dim(\bP_\lambda) &=& \frac{k!}{\prod_{(i,j)\in\lambda}h(i,j)}.
\end{eqnarray}
\end{thrm}

In \cite{zhang2014arxiv}, Schur-Weyl duality is employed to give a
computation about the integral of the following form:
\begin{eqnarray}
\int_{\unitary{d}} U^{\ot k}M (U^{\ot k})^\dagger
\dif\mu_{\mathrm{Haar}}(U).
\end{eqnarray}
Moreover we have obtained that
\begin{eqnarray}
\int_{\unitary{d}} U^{\ot k}M (U^{\ot
k})^\dagger\dif\mu_{\mathrm{Haar}}(U) = \Pa{\sum_{\pi\in
S_k}\Tr{MP(\pi)}P(\pi^{-1})}\Pa{\sum_{\pi\in
S_k}\mathrm{Wg}(\pi)P(\pi^{-1})},
\end{eqnarray}
where \emph{Weingarten function} $\mathrm{Wg}$ is defined over $S_k$
by
\begin{eqnarray}\label{eq:weingarten}
\mathrm{Wg}(\pi) := \frac1{(k!)^2}\sum_{\lambda\vdash
k}\frac{\dim(\bP_\lambda)^2}{\dim(\bQ_\lambda)}\chi_\lambda(\pi)
\end{eqnarray}
for each $\pi\in S_k$ and $\chi_\lambda(\pi) = \Tr{P_\lambda(\pi)}$
is the value of the character of irrep $\bP_\lambda$ at $\pi\in
S_k$.

Here we consider a special case where the above-mentioned
$M=\Lambda^{\ot k}$ for a given spectrum $\Lambda$ and any natural
number $k$, thus we introduce a new symbol for convenience:
\begin{eqnarray}
\sE_k(\Lambda):= \int (U\Lambda U^\dagger)^{\ot k}
\dif\mu_{\mathrm{Haar}}(U).
\end{eqnarray}
Throughout this paper, we frequently leave out the integral domain
$\unitary{d}$ when we consider matrix integral taken over the whole
unitary group $\unitary{d}$ unless stated otherwise. We see that
\begin{eqnarray}
\sE_k(\Lambda) =\sum_{\lambda\vdash(k,d)}\frac{\Tr{\Lambda^{\ot
k}C_\lambda}}{\Tr{C_\lambda}}C_\lambda,
\end{eqnarray}
where
\begin{eqnarray}
C_\lambda:=\frac{\dim(\bP_\lambda)}{k!}\sum_{\pi\in
S_k}\chi_\lambda(\pi)P(\pi).
\end{eqnarray}

\subsection{The case where $k=1,2$.}

It is already known in \cite{zhang2014arxiv} that
\begin{eqnarray}
\int UXU^\dagger \dif\mu_{\mathrm{Haar}}(U) = \frac{\Tr{X}}d\I_d
\end{eqnarray}
and
\begin{eqnarray}
&&\int (U\ot U)M(U\ot U)^\dagger \dif\mu_{\mathrm{Haar}}(U)\notag\\
&&= \Pa{\frac{\Tr{M}}{d^2-1} - \frac{\Tr{MF}}{d(d^2-1)}}\I_d\ot\I_d
- \Pa{\frac{\Tr{M}}{d(d^2-1)} - \frac{\Tr{MF}}{d^2-1}}F,
\end{eqnarray}
where $F:=\sum_{i,j=1}^d\out{ji}{ij}$ is called a swap operator.
Thus
\begin{eqnarray}
\sE_1(\Lambda)&=&\frac{\I_d}d,\label{e1}\\
\sE_2(\Lambda)&=&\frac1{d^2-1}\Br{\Pa{1-\frac{\Tr{\Lambda^2}}d}\I_d\ot\I_d
-\Pa{\frac1d-\Tr{\Lambda^2}}F}\notag\\
&=&\Delta^{(2)}_2C_{(2)}+\Delta^{(1,1)}_2C_{(1,1)},
\label{e2}\end{eqnarray} where
\begin{eqnarray}
\Delta^{(2)}_2:=\frac{1+\Tr{\Lambda^2}}{d(d+1)},\quad
\Delta^{(1,1)}_2:=\frac{1-\Tr{\Lambda^2}}{(d-1)d}
\end{eqnarray}
and
\begin{eqnarray}
C_\lambda =
\begin{cases}
\frac12(P_{(1)}+P_{(12)}),&\text{if}~\lambda=(2),\\
\frac12(P_{(1)}-P_{(12)}),&\text{if}~\lambda=(1,1).
\end{cases}
\end{eqnarray}

\subsection{The formula of $\sE_3(\Lambda)$}
In what follows, we compute $\sE_3(\Lambda)$. Note that we get the
following decomposition via Schur-Weyl duality
\begin{eqnarray}
(\complex^d)^{\ot3}\cong
\bQ_{(3)}\ot\bP_{(3)}\bigoplus\bQ_{(2,1)}\ot\bP_{(2,1)}\bigoplus\bQ_{(1,1,1)}\ot\bP_{(1,1,1)}
\end{eqnarray}
where
\begin{eqnarray}
\dim(\bQ_\lambda)=
\begin{cases}
\frac{d(d+1)(d+2)}6,&\text{if}~
\lambda=(3),\\
\frac{(d-1)d(d+1)}3, &\text{if}~ \lambda=(2,1),\\
\frac{(d-2)(d-1)d}6,&\text{if}~ \lambda=(1,1,1),
\end{cases}~\text{and}~\dim(\bP_\lambda)=
\begin{cases}
1,&\text{if}~
\lambda=(3),\\
2, &\text{if}~ \lambda=(2,1),\\
1,&\text{if}~ \lambda=(1,1,1).
\end{cases}
\end{eqnarray}
Hence
\begin{eqnarray}
C_\lambda =
\begin{cases}
\frac16\Pa{P_{(1)}+P_{(12)}+P_{(13)}+P_{(23)}+P_{(123)}+P_{(132)}},&\text{if}~
\lambda=(3),\\
\frac13\Pa{2P_{(1)}-P_{(123)}-P_{(132)}},&\text{if}~\lambda=(2,1),\\
\frac16\Pa{P_{(1)}-P_{(12)}-P_{(13)}-P_{(23)}+P_{(123)}+P_{(132)}},&\text{if}~\lambda=(1,1,1).
\end{cases}
\end{eqnarray}
It follows that
\begin{eqnarray}
\Tr{C_\lambda} =
\begin{cases}
\frac{d(d+1)(d+2)}6,&\text{if}~
\lambda=(3),\\
\frac{2(d-1)d(d+1)}3,&\text{if}~\lambda=(2,1),\\
\frac{(d-2)(d-1)d}6,&\text{if}~\lambda=(1,1,1)
\end{cases}
\end{eqnarray}
and
\begin{eqnarray}
\Tr{\Lambda^{\ot 3}C_\lambda} =
\begin{cases}
\frac16\Br{1+3\Tr{\Lambda^2}+2\Tr{\Lambda^3}},&\text{if}~
\lambda=(3),\\
\frac23\Br{1-\Tr{\Lambda^3}},&\text{if}~\lambda=(2,1),\\
\frac16\Br{1-3\Tr{\Lambda^2}+2\Tr{\Lambda^3}},&\text{if}~\lambda=(1,1,1).
\end{cases}
\end{eqnarray}

Therefore
\begin{eqnarray}
\sE_3(\rho) = \Delta^{(3)}_3C_{(3)} + \Delta^{(2,1)}_3C_{(2,1)}
+\Delta^{(1,1,1)}_3C_{(1,1,1)},
\end{eqnarray}
where
\begin{eqnarray}
\Delta^{(3)}_3 &:=& \frac{1+3\Tr{\Lambda^2}+2\Tr{\Lambda^3}}{d(d+1)(d+2)},\\
\Delta^{(2,1)}_3 &:=&\frac{1-\Tr{\Lambda^3}}{(d-1)d(d+1)},\\
\Delta^{(1,1,1)}_3
&:=&\frac{1-3\Tr{\Lambda^2}+2\Tr{\Lambda^3}}{(d-2)(d-1)d}.
\end{eqnarray}

\subsection{The formula of $\sE_4(\Lambda)$}

Similar we get the following decomposition:
\begin{eqnarray}
(\complex^d)^{\ot 4}&\cong& \bQ_{(4)}\ot\bP_{(4)}\bigoplus
\bQ_{(3,1)}\ot\bP_{(3,1)}\bigoplus
\bQ_{(2,2)}\ot\bP_{(2,2)}\nonumber\\
&&\bigoplus \bQ_{(2,1,1)}\ot\bP_{(2,1,1)}\bigoplus
\bQ_{(1,1,1,1)}\ot\bP_{(1,1,1,1)},
\end{eqnarray}
where
\begin{eqnarray}
\dim(\bQ_\lambda)=
\begin{cases}
\frac{d(d+1)(d+2)(d+3)}{24},&\text{if}~
\lambda=(4),\\
\frac{(d-1)d(d+1)(d+2)}8, &\text{if}~ \lambda=(3,1),\\
\frac{(d-1)d^2(d+1)}{12},&\text{if}~
\lambda=(2,2),\\
\frac{(d-2)(d-1)d(d+1)}8,&\text{if}~ \lambda=(2,1,1),\\
\frac{(d-3)(d-2)(d-1)d}{24},&\text{if}~ \lambda=(1,1,1,1),
\end{cases}~\text{and}\quad\dim(\bP_\lambda)=
\begin{cases}
1,&\text{if}~
\lambda=(4),\\
3, &\text{if}~ \lambda=(3,1),\\
2,&\text{if}~
\lambda=(2,2),\\
3,&\text{if}~ \lambda=(2,1,1),\\
1,&\text{if}~ \lambda=(1,1,1,1).
\end{cases}
\end{eqnarray}
Hence we have:
\begin{eqnarray}
C_{(4)} &=&
\frac1{24}P_{(1)}+\frac1{24}\Pa{P_{(12)}+P_{(13)}+P_{(14)}+P_{(23)}+P_{(24)}+P_{(34)}}\notag\\
&&+\frac1{24}\Pa{P_{(12)(34)}+P_{(13)(24)}+P_{(14)(23)}}\notag\\
&&+\frac1{24}\Pa{P_{(123)}+P_{(132)}+P_{(124)}+P_{(142)}+P_{(134)}+P_{(143)}+P_{(234)}+P_{(243)}}\notag\\
&&+\frac1{24}\Pa{P_{(1234)}+P_{(1243)}+P_{(1324)}+P_{(1342)}+P_{(1423)}+P_{(1432)}},
\end{eqnarray}
\begin{eqnarray}
C_{(3,1)} &=&
\frac38P_{(1)}+\frac18\Pa{P_{(12)}+P_{(13)}+P_{(14)}+P_{(23)}+P_{(24)}+P_{(34)}}\notag\\
&&-\frac18\Pa{P_{(12)(34)}+P_{(13)(24)}+P_{(14)(23)}}\notag\\
&&-\frac18\Pa{P_{(1234)}+P_{(1243)}+P_{(1324)}+P_{(1342)}+P_{(1423)}+P_{(1432)}},
\end{eqnarray}
\begin{eqnarray}
C_{(2,2)} &=& \frac16P_{(1)}+\frac16\Pa{P_{(12)(34)}+P_{(13)(24)}+P_{(14)(23)}}\notag\\
&&-\frac1{12}\Pa{P_{(123)}+P_{(132)}+P_{(124)}+P_{(142)}+P_{(134)}+P_{(143)}+P_{(234)}+P_{(243)}},
\end{eqnarray}
\begin{eqnarray}
C_{(2,1,1)} &=&
\frac38P_{(1)}-\frac18\Pa{P_{(12)}+P_{(13)}+P_{(14)}+P_{(23)}+P_{(24)}+P_{(34)}}\notag\\
&&-\frac18\Pa{P_{(12)(34)}+P_{(13)(24)}+P_{(14)(23)}}\notag\\
&&+\frac18\Pa{P_{(1234)}+P_{(1243)}+P_{(1324)}+P_{(1342)}+P_{(1423)}+P_{(1432)}},
\end{eqnarray}
\begin{eqnarray}
C_{(1,1,1,1)} &=&
\frac1{24}P_{(1)}-\frac1{24}\Pa{P_{(12)}+P_{(13)}+P_{(14)}+P_{(23)}+P_{(24)}+P_{(34)}}\notag\\
&&+\frac1{24}\Pa{P_{(12)(34)}+P_{(13)(24)}+P_{(14)(23)}}\notag\\
&&+\frac1{24}\Pa{P_{(123)}+P_{(132)}+P_{(124)}+P_{(142)}+P_{(134)}+P_{(143)}+P_{(234)}+P_{(243)}}\notag\\
&&-\frac1{24}\Pa{P_{(1234)}+P_{(1243)}+P_{(1324)}+P_{(1342)}+P_{(1423)}+P_{(1432)}},
\end{eqnarray}
\begin{eqnarray}
\sE_4(\rho) = \Delta^{(4)}_4C_{(4)}+\Delta^{(3,1)}_4C_{(3,1)}
+\Delta^{(2,2)}_4C_{(2,2)}+\Delta^{(2,1,1)}_4C_{(2,1,1)}
+\Delta^{(1,1,1,1)}_4C_{(1,1,1,1)},
\end{eqnarray}
where
\begin{eqnarray}
\Delta^{(4)}_4&:=&\frac{1+6\Tr{\Lambda^2}+3\Tr{\Lambda^2}^2
+8\Tr{\Lambda^3}+6\Tr{\Lambda^4}}{d(d+1)(d+2)(d+3)},\label{eq:delta4}\\
\Delta^{(3,1)}_4&:=&\frac{1+2\Tr{\Lambda^2}-\Tr{\Lambda^2}^2-2\Tr{\Lambda^4}}{(d-1)d(d+1)(d+2)},\label{eq:delta31}\\
\Delta^{(2,2)}_4&:=&\frac{1+3\Tr{\Lambda^2}^2-4\Tr{\Lambda^3}}{(d-1)d^2(d+1)},\label{eq:delta211}\\
\Delta^{(2,1,1)}_4&:=&\frac{1-2\Tr{\Lambda^2}-\Tr{\Lambda^2}^2+2\Tr{\Lambda^4}}{(d-2)(d-1)d(d+1)},\\
\Delta^{(1,1,1,1)}_4&:=&\frac{1-6\Tr{\Lambda^2}+3\Tr{\Lambda^2}^2+8\Tr{\Lambda^3}-6\Tr{\Lambda^4}}{(d-3)(d-2)(d-1)d}.\label{eq:delta1111}
\end{eqnarray}

\subsubsection{The $(k,d)=(2,2)$ case}

We have
\begin{eqnarray}
\Delta^{(2)}_2=\frac{1+t_2}6,\quad \Delta^{(1,1)}_2=\frac{1-t_2}2.
\end{eqnarray}

\subsubsection{The $(k,d)=(3,2)$ case}

We have
\begin{eqnarray}
\Delta^{(3)}_3 = \frac{1+3t_2+2t_3}{24},\quad \Delta^{(2,1)}_3
:=\frac{1-t_3}6,\quad \Delta^{(1,1,1)}_3=0.
\end{eqnarray}

\subsubsection{The $(k,d)=(4,2)$ case}

We have
\begin{eqnarray}
\Delta^{(4)}_4&=&\frac{1+6t_2+3t^2_2
+8t_3+6t_4}{120},\\
\Delta^{(3,1)}_4&=&\frac{1+2t_2-t^2_2-2t_4}{24},\\
\Delta^{(2,2)}_4&=&\frac{1+3t^2_2-4t_3}{12},\\
\Delta^{(1,1,1,1)}_4&=& 0.
\end{eqnarray}

\subsection{The moment of $\Tr{\rho^k}=t_k$}

In fact, we have already known that
\begin{prop}[\cite{Zyczkowski2001jpa}]\label{prop:karol}
We have:
\begin{eqnarray}
\label{7}\langle t_2\rangle=\int \dif\mu_{\rH\rS}(\rho)\Tr{\rho^2} &=& \frac{2d}{d^2+1},\\
\label{8}\langle t_3\rangle=\int \dif\mu_{\rH\rS}(\rho)\Tr{\rho^3}
&=&
\frac{5d^2+1}{\Pa{d^2+1}\Pa{d^2+2}},\\
\label{9}\langle t_4\rangle=\int \dif\mu_{\rH\rS}(\rho)\Tr{\rho^4}
&=& \frac{14d^3+10d}{\Pa{d^2+1}\Pa{d^2+2}\Pa{d^2+3}}.
\end{eqnarray}
\end{prop}

\begin{remark}
It is obvious that\begin{eqnarray}
\label{e7}\langle t_2\rangle=\int \dif\nu(\Lambda)\Tr{\Lambda^2} &=& \frac{2d}{d^2+1},\\
\label{e8}\langle t_3\rangle=\int \dif\nu(\Lambda)\Tr{\Lambda^3} &=&
\frac{5d^2+1}{\Pa{d^2+1}\Pa{d^2+2}},\\
\label{e9}\langle t_4\rangle=\int \dif\nu(\Lambda)\Tr{\Lambda^4} &=&
\frac{14d^3+10d}{\Pa{d^2+1}\Pa{d^2+2}\Pa{d^2+3}}.
\end{eqnarray}
\end{remark}

\begin{lem}\label{lem:lin}
It holds that
\begin{eqnarray}
\label{10}\left\langle t^2_2\right\rangle = \int
\dif\mu_{\rH\rS}(\rho)\Br{\Tr{\rho^2}}^2
=\frac{4d^4+18d^2+2}{\Pa{d^2+1}\Pa{d^2+2}\Pa{d^2+3}}.
\end{eqnarray}
\end{lem}

\begin{proof}
In what follows, we calculate the integral:
\begin{eqnarray}
&&\int \dif\mu_{\rH\rS}(\rho)\Br{\Tr{\rho^2}}^2=
\int\dif\nu(\Lambda)\Tr{\Lambda^2}^2 = \int\dif\nu(\Lambda)\Pa{\Tr{\Lambda^4}+2\sum_{i<j}\lambda^2_i\lambda^2_j}\notag\\
&&= \int \dif\nu(\Lambda)\Tr{\Lambda^4}+
2C^d_{\rH\rS}\int\Pa{\sum_{i<j}\lambda^2_i\lambda^2_j}
\delta\Pa{1-\sum^d_{j=1}\lambda_j}\abs{\Delta(\lambda)}^2\prod^d_{j=1}\dif\lambda_j\notag\\
&&=\int
\dif\nu(\Lambda)\Tr{\Lambda^4}+2C^d_{\rH\rS}\binom{d}{2}\int\Pa{\lambda^2_1\lambda^2_2}
\delta\Pa{1-\sum^d_{j=1}\lambda_j}\abs{\Delta(\lambda)}^2\prod^d_{j=1}\dif\lambda_j,
\end{eqnarray}
where $C^d_{\rH\rS}$ is the normalization constant:
\begin{eqnarray}
C^d_{\rH\rS} =
\frac{\Gamma\Pa{d^2}}{\Gamma(d+1)\prod^d_{j=1}\Gamma(j)^2}.
\end{eqnarray}

Next, we calculate the following integral:
\begin{eqnarray}
\int\Pa{\lambda^2_1\lambda^2_2}
\delta\Pa{1-\sum^d_{j=1}\lambda_j}\abs{\Delta(\lambda)}^2\prod^d_{j=1}\dif\lambda_j.
\end{eqnarray}
Let
\begin{eqnarray}
F(t) = \int\Pa{\lambda^2_1\lambda^2_2}
\delta\Pa{t-\sum^d_{j=1}\lambda_j}\abs{\Delta(\lambda)}^2\prod^d_{j=1}\dif\lambda_j.
\end{eqnarray}
Performing Laplace transform $(t\to s)$ of $F(t)$ gives rise to
\begin{eqnarray}
\widetilde F(s)&=& \int^\infty_0
\Pa{\lambda^2_1\lambda^2_2}\exp\Pa{-s\sum^d_{j=1}\lambda_j}\abs{\Delta(\lambda)}^2\prod^d_{j=1}\dif\lambda_j\notag\\
&=&s^{-(d^2+4)}\int^\infty_0
\Pa{\mu^2_1\mu^2_2}\exp\Pa{-\sum^d_{j=1}\mu_j}\abs{\Delta(\mu)}^2\prod^d_{j=1}\dif\mu_j.
\end{eqnarray}
Using the inverse Laplace transform result $(s\to t)$:
$\sL^{-1}(s^\alpha)=\frac{t^{-\alpha-1}}{\Gamma(-\alpha)}$, it
follows that
\begin{eqnarray}
F(t) = \frac1{\Gamma(d^2+4)}t^{d^2+3}\int^\infty_0
\Pa{\mu^2_1\mu^2_2}\exp\Pa{-\sum^d_{j=1}\mu_j}\abs{\Delta(\mu)}^2\prod^d_{j=1}\dif\mu_j.
\end{eqnarray}
Then
\begin{eqnarray}
\int\Pa{\lambda^2_1\lambda^2_2}
\delta\Pa{1-\sum^d_{j=1}\lambda_j}\abs{\Delta(\lambda)}^2\prod^d_{j=1}\dif\lambda_j
= \frac1{\Gamma(d^2+4)}\int^\infty_0
\Pa{\mu^2_1\mu^2_2}\exp\Pa{-\sum^d_{j=1}\mu_j}\abs{\Delta(\mu)}^2\prod^d_{j=1}\dif\mu_j.
\end{eqnarray}
Denote
\begin{eqnarray}
\langle f(\mu)\rangle_q = \frac{\int^\infty_0\cdots \int^\infty_0
f(\mu)q(\mu)\dif\mu}{\int^\infty_0\cdots\int^\infty_0q(\mu)\dif\mu},
\end{eqnarray}
where
\begin{eqnarray}
q(\mu)\equiv
q(\mu_1,\ldots,\mu_d)=\abs{\Delta(\mu)}^{2\gamma}\prod^d_{j=1}\mu^{\alpha-1}_je^{-\mu_j}\dif
\mu_j.
\end{eqnarray}
From Mehta's book \cite[Eq.~(17.8.3), pp.~324]{mehta2004elsevier},
we see that
\begin{eqnarray}
\langle \mu^2_1\cdots \mu^2_k \mu_{k+1}\cdots \mu_m\rangle_q
=\prod^k_{j=1}(\alpha+1+\gamma(2d-m-j))\prod^m_{j=1}(\alpha+\gamma(d-j)).
\end{eqnarray}
Letting $k=m=2$ and $(\alpha,\gamma)=(1,1)$ in the above equation,
we obtain that
\begin{eqnarray}
\langle \mu^2_1\mu^2_2\rangle_q
=\prod^2_{j=1}(2d-j)\cdot\prod^2_{j=1}(d-j+1) = 2d(d-1)^2(2d-1).
\end{eqnarray}
This implies that
\begin{eqnarray}
\int^\infty_0
\Pa{\mu^2_1\mu^2_2}\exp\Pa{-\sum^d_{j=1}\mu_j}\abs{\Delta(\mu)}^2\prod^d_{j=1}\dif\mu_j
=2d(d-1)^2(2d-1)\int^\infty_0\exp\Pa{-\sum^d_{j=1}\mu_j}\abs{\Delta(\mu)}^2\prod^d_{j=1}\dif\mu_j.
\end{eqnarray}
Therefore
\begin{eqnarray}
\int\Pa{\lambda^2_1\lambda^2_2}
\delta\Pa{1-\sum^d_{j=1}\lambda_j}\abs{\Delta(\lambda)}^2\prod^d_{j=1}\dif\lambda_j
=
\frac{2d(d-1)^2(2d-1)}{\Gamma(d^2+4)}\int^\infty_0\exp\Pa{-\sum^d_{j=1}\mu_j}\abs{\Delta(\mu)}^2\prod^d_{j=1}\dif\mu_j.
\end{eqnarray}
Since
\begin{eqnarray}
\int^\infty_0\cdots\int^\infty_0 q(\mu)\dif\mu=
\prod^d_{j=1}\Gamma(j)\Gamma(j+1) =
\Gamma(d+1)\prod^d_{j=1}\Gamma(j)^2.
\end{eqnarray}
Finally we get
\begin{eqnarray}
\int\Pa{\lambda^2_1\lambda^2_2}
\delta\Pa{1-\sum^d_{j=1}\lambda_j}\abs{\Delta(\lambda)}^2\prod^d_{j=1}\dif\lambda_j
=2d(d-1)^2(2d-1)\frac{\Gamma(d+1)}{\Gamma(d^2+4)}\prod^d_{j=1}\Gamma(j)^2.
\end{eqnarray}
Based on this computation, we finally obtain that
\begin{eqnarray}
\int \dif\mu_{\rH\rS}(\rho)\Br{\Tr{\rho^2}}^2 &=&
\frac{14d^3+10d}{\Pa{d^2+1}\Pa{d^2+2}\Pa{d^2+3}} \notag\\
&&+2
\frac{\Gamma\Pa{d^2}}{\Gamma(d+1)\prod^d_{j=1}\Gamma(j)^2}\binom{d}{2}2d(d-1)^2(2d-1)\frac{\Gamma(d+1)}{\Gamma(d^2+4)}\prod^d_{j=1}\Gamma(j)^2\notag\\
&=&\frac{14d^3+10d}{\Pa{d^2+1}\Pa{d^2+2}\Pa{d^2+3}}+
\frac{2(d-1)^3(2d-1)}{\Pa{d^2+1}\Pa{d^2+2}\Pa{d^2+3}}.
\end{eqnarray}
Therefore we completes the proof.
\end{proof}

\subsection{The proof of Theorem~\ref{th:iso}}\label{subsect:iso}

For the first term in the left hands (lhs) of the above equation:
\begin{eqnarray}
\Tr{\Br{A^2\ot B^2}\sE_2(\Lambda)} &=& \frac{\Delta^{(2)}_2}2\Br{\Tr{A^2}\Tr{B^2}+\Tr{A^2B^2}}\notag\\
 &&+ \frac{\Delta^{(1,1)}_2}2\Br{\Tr{A^2}\Tr{B^2}-\Tr{A^2B^2}}.
\end{eqnarray}
Then for the third and fourth terms:
\begin{eqnarray}
&&\Tr{\Br{A^2\ot B^{\ot 2}}\sE_3(\Lambda)}\notag \\
&&= \frac{\Delta^{(3)}_3}6\Br{\Tr{A^2}\Tr{B}^2 +
2\Tr{A^2B}\Tr{B}+\Tr{A^2}\Tr{B^2}+2\Tr{A^2B^2}}\notag
\\
&&~~~+ \frac{2\Delta^{(2,1)}_3}3\Br{\Tr{A^2}\Tr{B}^2 -
\Tr{A^2B^2}}\notag\\
&&~~~+\frac{\Delta^{(1,1,1)}_3}6\Br{\Tr{A^2}\Tr{B}^2 -
2\Tr{A^2B}\Tr{B}-\Tr{A^2}\Tr{B^2}+2\Tr{A^2B^2}}
\end{eqnarray}
and
\begin{eqnarray}
&&\Tr{\Br{B^2\ot A^{\ot 2}}\sE_3(\Lambda)} \notag\\
&&= \frac{\Delta^{(3)}_3}6\Br{\Tr{B^2}\Tr{A}^2 +
2\Tr{B^2A}\Tr{A}+\Tr{B^2}\Tr{A^2}+2\Tr{B^2A^2}}\notag
\\
&&~~+ \frac{2\Delta^{(2,1)}_3}3\Br{\Tr{B^2}\Tr{A}^2 -
\Tr{B^2A^2}}\notag\\
&&~~+\frac{\Delta^{(1,1,1)}_3}6\Br{\Tr{B^2}\Tr{A}^2 -
2\Tr{B^2A}\Tr{A}-\Tr{B^2}\Tr{A^2}+2\Tr{B^2A^2}}.
\end{eqnarray}
The second term is:
\begin{eqnarray}
&&\Tr{\Br{A^{\ot 2}\ot B^{\ot 2}}\sE_4(\Lambda)} \notag\\
&&= \frac{\Delta^{(4)}_4}{24}\left[\Tr{A}^2\Tr{B}^2 + \Tr{A^2}\Tr{B}^2 + 4\Tr{AB}\Tr{A}\Tr{B}+\Tr{A}^2\Tr{B^2}+\Tr{A^2}\Tr{B^2}\right.\notag\\
&&~~~~~~~~~~~\left.+2\Tr{AB}^2+ 4\Tr{A^2B}\Tr{B}+ 4\Tr{A}\Tr{AB^2}+4\Tr{A^2B^2}+2\Tr{ABAB}\right]\notag\\
&&~~+\frac{\Delta^{(3,1)}_4}8\left[3\Tr{A}^2\Tr{B}^2 + \Tr{A^2}\Tr{B}^2 + 4\Tr{AB}\Tr{A}\Tr{B}+\Tr{A}^2\Tr{B^2}\right.\notag\\
&&~~~~~~~~~~~~\left.-\Tr{A^2}\Tr{B^2}-2\Tr{AB}^2- 4\Tr{A^2B^2}-2\Tr{ABAB}\right]\notag\\
&&~~+\frac{\Delta^{(2,2)}_4}{12}\Br{2\Tr{A}^2\Tr{B}^2+2\Tr{A^2}\Tr{B^2}+4\Tr{AB}^2-4\Tr{A^2B}\Tr{B}-4\Tr{A}\Tr{AB^2}}\notag\\
&&~~+\frac{\Delta^{(2,1,1)}_4}8\left[3\Tr{A}^2\Tr{B}^2- \Tr{A^2}\Tr{B}^2 - 4\Tr{AB}\Tr{A}\Tr{B}-\Tr{A}^2\Tr{B^2}\right.\notag\\
&&~~~~~~~~~~~~\left.-\Tr{A^2}\Tr{B^2}-2\Tr{AB}^2+4\Tr{A^2B^2}+2\Tr{ABAB}\right]\notag\\
&&~~+\frac{\Delta^{(1,1,1,1)}_4}{24}\left[\Tr{A}^2\Tr{B}^2 - \Tr{A^2}\Tr{B}^2 - 4\Tr{AB}\Tr{A}\Tr{B}-\Tr{A}^2\Tr{B^2}+\Tr{A^2}\Tr{B^2}\right.\notag\\
&&~~~~~~~~~~~\left.+2\Tr{AB}^2+ 4\Tr{A^2B}\Tr{B}+
4\Tr{A}\Tr{AB^2}-4\Tr{A^2B^2}-2\Tr{ABAB}\right].
\end{eqnarray}
Therefore, we get the conclusion.

\subsection{The proof of Theorem~\ref{th:pure-case}}\label{subsect:pure-case}

Clearly, for $k=2,3,4$, we know that
\begin{eqnarray}\label{eq:symmetric-proj}
\int\dif\mu(\psi)\out{\psi}{\psi}^{\ot k} = \Delta^{(k)}_k C_{(k)}.
\end{eqnarray}
Then
\begin{eqnarray}
&&\int\Delta A(\psi)^2\cdot\Delta B(\psi)^2 \dif\mu(\psi)\notag \\
&&= \int\dif\mu(\psi)\Tr{\Br{A^2\ot B^2}\psi^{\ot 2}} +
\int\dif\mu(\psi)\Tr{\Br{A^{\ot 2}\ot B^{\ot 2}}\psi^{\ot 4}} \notag\\
&&~~~- \int\dif\mu(\psi)\Tr{\Br{A^2\ot B^{\ot 2}}\psi^{\ot 3}} -
\int\dif\mu(\psi)\Tr{\Br{B^2\ot A^{\ot 2}}\psi^{\ot 3}},
\end{eqnarray}
Now,
\begin{eqnarray}
\int\Tr{\Br{A^2\ot B^2}\psi^{\ot 2}}\dif\mu(\psi) =
\frac1{d(d+1)}\Br{\Tr{A^2}\Tr{B^2}+\Tr{A^2B^2}}
\end{eqnarray}
\begin{eqnarray}
\int\Tr{\Br{A^{\ot 2}\ot B^{\ot 2}}\psi^{\ot 4}} \dif\mu(\psi)=
\frac1{d(d+1)(d+2)(d+3)}\Omega(A,B),
\end{eqnarray}
where
\begin{eqnarray}
\Omega(A,B)&=&\Tr{A}^2\Tr{B}^2 + \Tr{A^2}\Tr{B}^2+ 4\Tr{AB}\Tr{A}\Tr{B} + \Tr{A}^2\Tr{B^2}\notag\\
&&+\Tr{A^2}\Tr{B^2} + 2\Tr{AB}^2 + 4\Tr{A^2B}\Tr{B}+ 4\Tr{A}\Tr{AB^2}\notag\\
&&+ 4\Tr{A^2B^2} + 2\Tr{ABAB}.
\end{eqnarray}
Moreover
\begin{eqnarray}
&&\int\Tr{\Br{A^2\ot B^{\ot 2}}\psi^{\ot 3}}\dif\mu(\psi) \notag\\
&&= \frac1{d(d+1)(d+2)}\Br{\Tr{A^2}\Tr{B}^2 + 2\Tr{A^2B}\Tr{B} +
\Tr{A^2}\Tr{B^2} + 2\Tr{A^2B^2}}
\end{eqnarray}
and
\begin{eqnarray}
&&\int\Tr{\Br{B^2\ot A^{\ot 2}}\psi^{\ot 3}}\dif\mu(\psi) \notag\\
&&= \frac1{d(d+1)(d+2)}\Br{\Tr{B^2}\Tr{A}^2 + 2\Tr{B^2A}\Tr{A} +
\Tr{A^2}\Tr{B^2} + 2\Tr{A^2B^2}}.
\end{eqnarray}
Thus
\begin{eqnarray}
&&d(d+1)(d+2)(d+3)\int\Delta A(\psi)^2\cdot\Delta B(\psi)^2\dif\mu(\psi)\notag\\
&& = \Omega(A,B) + (d+2)(d+3)\Br{\Tr{A^2}\Tr{B^2}+\Tr{A^2B^2}}\notag\\
&&~~~~- (d+3)\Br{\Tr{A^2}\Tr{B}^2 + 2\Tr{A^2B}\Tr{B} +
\Tr{A^2}\Tr{B^2}
+ 2\Tr{A^2B^2}}\notag\\
&&~~~~- (d+3)\Br{\Tr{B^2}\Tr{A}^2 + 2\Tr{B^2A}\Tr{A} +
\Tr{A^2}\Tr{B^2} + 2\Tr{A^2B^2}}
\end{eqnarray}

\begin{eqnarray}
&&d(d+1)(d+2)(d+3)\int\Delta A(\psi)^2\cdot\Delta B(\psi)^2\dif\mu(\psi)\notag\\
&& = \Tr{A}^2\Tr{B}^2 -(d+2) \Tr{A^2}\Tr{B}^2+ 4\Tr{AB}\Tr{A}\Tr{B} \notag\\
&&~~~~-(d+2)\Tr{A}^2\Tr{B^2} +(d^2+3d+1)\Tr{A^2}\Tr{B^2} + 2\Tr{AB}^2 \notag\\
&&~~~~- 2(d+1)\Tr{A^2B}\Tr{B}- 2(d+1)\Tr{A}\Tr{AB^2}\notag\\
&&~~~~+ (d^2+d-2)\Tr{A^2B^2} + 2\Tr{ABAB}.
\end{eqnarray}
Therefore
\begin{eqnarray}
&&d(d+1)(d+2)(d+3)\int\Delta A(\psi)^2\cdot\Delta B(\psi)^2\dif\mu(\psi)\notag\\
&& = \Tr{A}^2\Tr{B}^2 + 4\Tr{AB}\Tr{A}\Tr{B}  + 2\Tr{AB}^2  + 2\Tr{ABAB}\notag\\
&&~~~~-(d+2) \Br{\Tr{A^2}\Tr{B}^2+\Tr{A}^2\Tr{B^2}} \notag\\
&&~~~~- 2(d+1)\Br{\Tr{A^2B}\Tr{B}+\Tr{A}\Tr{AB^2}}\notag\\
&&~~~~+(d^2+3d+1)\Tr{A^2}\Tr{B^2}+ (d^2+d-2)\Tr{A^2B^2},
\end{eqnarray}
that is,
\begin{eqnarray}
\int\Delta A(\psi)^2\cdot\Delta B(\psi)^2\dif\mu(\psi) =
\sum^8_{j=1}u_j\Omega_j(A,B),
\end{eqnarray}
where $\Omega_j(A,B)$ is from Theorem~\ref{th:iso}, and for
$K_d=(d(d+1)(d+2)(d+3))^{-1}$,
\begin{eqnarray}
&&u_1 = K_d,\quad u_2=-(d+2)K_d,\quad u_3 =4K_d,\quad u_4=(d^2+3d+1)K_d,\\
&&u_5 = 2K_d, \quad u_6=-2(d+1)K_d,\quad u_7=(d^2+d-2)K_d, \quad
u_8=2K_d.
\end{eqnarray}
In the following we calculate the average lower bound,
\begin{eqnarray}
&&\int\dif\mu(\psi)\Pa{\langle\set{A,B}\rangle_\psi-\langle
A\rangle_\psi\langle
B\rangle_\psi}^2=\int\dif\mu(\psi)\Br{\langle\set{A,B}\rangle^2_\psi+\langle
A\rangle^2_\psi\langle B\rangle^2_\psi -
2\langle\set{A,B}\rangle_\psi\langle A\rangle_\psi\langle
B\rangle_\psi}\notag\\
&&=\frac1{d(d+1)}\Br{\Tr{\set{A,B}}^2 + \Tr{\set{A,B}^2}} +
\frac1{d(d+1)(d+2)(d+3)}\Omega(A,B)\notag\\
&&~~~-
\frac2{d(d+1)(d+2)}[\Tr{A}\Tr{B}\Tr{\set{A,B}}+\Tr{AB}\Tr{\set{A,B}}
+ \Tr{A\set{A,B}}\Tr{B}  \notag\\
&&~~~+\Tr{A}\Tr{B\set{A,B}} + \Tr{AB\set{A,B}} + \Tr{A\set{A,B}B}]\notag\\
&&=\frac1{d(d+1)}\Br{\Tr{AB}^2 +
\frac12\Tr{ABAB}+\frac12\Tr{A^2B^2}} +
\frac1{d(d+1)(d+2)(d+3)}\Omega(A,B)\notag\\
&&~~~- \frac2{d(d+1)(d+2)}[\Tr{A}\Tr{B}\Tr{AB}+\Tr{AB}^2
+ \Tr{A^2B}\Tr{B}  \notag\\
&&~~~\left.+\Tr{B^2A}\Tr{A} + 2\Tr{ABAB} + 2\Tr{A^2B^2}\right]
\end{eqnarray}
and
\begin{eqnarray}
\int\langle [A,B]\rangle^2_\psi\dif\mu(\psi) &=&
\int\dif\mu(\psi)\Br{\Innerm{\psi}{AB}{\psi}^2 +
\Innerm{\psi}{BA}{\psi}^2 -
2\Innerm{\psi}{AB}{\psi}\Innerm{\psi}{BA}{\psi}}\notag\\
&=& \frac1{2d(d+1)}\Br{\Tr{A^2B^2}-\Tr{ABAB}}.
\end{eqnarray}
Therefore we have that
\begin{eqnarray}
\int\dif\mu(\psi)\Br{\Pa{\langle\set{A,B}\rangle_\psi-\langle
A\rangle_\psi\langle B\rangle_\psi}^2 + \langle [A,B]\rangle^2_\psi}
= \sum^8_{j=1}l_j\Omega_j(A,B),
\end{eqnarray}
where
\begin{eqnarray}
&&l_1 = K_d,\quad l_2 = K_d,\quad l_3 = -2(d+1)K_d,\quad l_4 = K_d,\quad l_5 = (d+1)(d+2)K_d, \\
&&\quad l_6 = -2(d+1)K_d,\quad l_7 = (d^2+d-2)K_d, \quad l_8 =
-2(2d+5)K_d.
\end{eqnarray}

\subsection{The proof of Theorem~\ref{th:ave-lower}}\label{subset:ave-lower}

Since the first term in the rhs of \eqref{eq:R-S} can be rewritten
as
\begin{eqnarray}
&&\Pa{\langle\set{A,B}\rangle_\rho-\langle A\rangle_\rho\langle
B\rangle_\rho}^2 = \langle\set{A,B}\rangle^2_\rho + \langle
A\rangle^2_\rho\langle B\rangle^2_\rho -
2\langle\set{A,B}\rangle_\rho\langle A\rangle_\rho\langle
B\rangle_\rho\notag\\
\label{12}&&= \Tr{\rho^{\ot 2}\set{A,B}^{\ot 2}} + \Tr{\rho^{\ot
4}\Br{A^{\ot2}\ot B^{\ot 2}}} - 2\Tr{\rho^{\ot 3}\Br{\set{A,B}\ot
A\ot B}},
\end{eqnarray}
it follows that
\begin{eqnarray}
\nonumber
\int_{\density{\cH_d}}\dif\mu_{\rH\rS}(\rho)\langle\set{A,B}\rangle^2_\rho
&=&\nonumber\int\dif\nu(\Lambda)\Tr{\sE_2(\Lambda)\set{A,B}^{\ot2}} \\
&=& \label{1}\frac1{d^2+1}\Tr{AB}^2 +
\frac1{2d(d^2+1)}\Br{\Tr{A^2B^2} + \Tr{ABAB}}
\end{eqnarray}
and
\begin{eqnarray}
\nonumber &&\int_{\density{\cH_d}}\dif\mu_{\rH\rS}(\rho)\langle
A\rangle^2_\rho\langle B\rangle^2_\rho =\nonumber
\int\dif\nu(\Lambda)\Tr{\sE_4(\Lambda)\Br{A^{\ot2}\ot B^{\ot 2}}}
\\
&&=\nonumber\alpha_1\Omega_1(A,B)+\alpha_2\Pa{\Omega_2(A,B)+4\Omega_3(A,B)}
+\alpha_3\Pa{\Omega_4(A,B)+2\Omega_5(A,B)}
\\\label{2}&&~~~+\alpha_4\Omega_6(A,B)+\alpha_5\Pa{2\Omega_7(A,B)+\Omega_8(A,B)},
\end{eqnarray}
where
\begin{eqnarray}
\alpha_1&=&N_d\Pa{d^4-18d^2+158-\frac{50}{d^2+1}+\frac{792}{d^2+2}-\frac{1512}{d^2+3}}
\\ \alpha_2&=&N_d\Pa{d^3-20d+\frac{50d}{d^2+1}-\frac{396d}{d^2+2}+\frac{504d}{d^2+3}},
\\ \alpha_3&=&N_d\Pa{-2d^2-20+\frac{50}{d^2+1}-\frac{396}{d^2+2}+\frac{504}{d^2+3}},
\\ \alpha_4&=&N_d\Pa{4d^2-380+\frac{500}{d^2+1}-\frac{1584}{d^2+2}+\frac{2016}{d^2+3}},
\\
\alpha_5&=&N_d\Pa{2d-\frac{100d}{d^2+1}+\frac{396d}{d^2+2}-\frac{336d}{d^2+3}}.
\end{eqnarray}
Moreover
\begin{eqnarray}
\Tr{\sE_3(\Lambda)\Br{\set{A,B}\ot A\ot
B}}&=&\frac{1}{6}(\Delta^{(3)}_3+4\Delta^{(2,1)}_3+\Delta^{(1,1,1)}_3)\Br{\Tr
{AB}\Tr A\Tr B}\notag
\\&&+\frac{1}{6}(\Delta^{(3)}_3-\Delta^{(1,1,1)}_3)\Br{\Tr{A^2B}\Tr B+\Tr{AB^2}\Tr
A+(\Tr{AB})^2}\notag
\\&&+\frac{1}{6}(\Delta^{(3)}_3-2\Delta^{(2,1)}_3+\Delta^{(1,1,1)}_3)\Br{\Tr{ABAB}+\Tr{A^2B^2}}.
\end{eqnarray}
Therefore
\begin{eqnarray}
\nonumber
&&\int_{\density{\cH_d}}\dif\mu_{\rH\rS}(\rho)\langle\set{A,B}\rangle_\rho\langle
A\rangle_\rho\langle B\rangle_\rho
=\int\dif\nu(\Lambda)\Tr{\sE_3(\Lambda)\Br{\set{A,B}\ot A\ot
B}}\\
\nonumber
&&=L_d\Pa{d^2-8-\frac{10}{d^2+1}+\frac{36}{d^2+2}}\Omega_3(A,B)
\\ \nonumber &&+L_d\Pa{1+\frac{10}{d^2+1}-\frac{18}{d^2+2}}\Br{\Omega_5(A,B)+\Omega_6(A,B)}
\\\label{3}&&+L_d\Pa{\frac{1}{2}d^4-2d^2+1+\frac{10}{d^2+1}-\frac{18}{d^2+2}}\Br{\Omega_7(A,B)+\Omega_8(A,B)}
\end{eqnarray}
where $L_d=d(d^2-1)(d^2-4).$

Since\begin{eqnarray}
\nonumber&&\int_{\density{\cH_d}}\dif\mu_{\rH\rS}(\rho)\Pa{\langle\set{A,B}\rangle_\rho-\langle
A\rangle_\rho\langle B\rangle_\rho}^2
\\&&=\int_{\density{\cH_d}}\dif\mu_{\rH\rS}(\rho)
\Pa{\langle\set{A,B}\rangle_\rho^2-2\langle\set{A,B}\rangle_\rho\langle
A\rangle_\rho\langle B\rangle_\rho+\langle A\rangle_\rho^2\langle
B\rangle_\rho^2},
\end{eqnarray}
and by \eqref{1}, \eqref{2} and \eqref{3}, we obtain the equality
\eqref{4}.

Moreover,
\begin{eqnarray}
\int_{\density{\cH_d}}\dif\mu_{\rH\rS}(\rho)\langle
[A,B]\rangle^2_\rho &=&
\int\dif\nu(\Lambda)\Tr{\sE_2(\Lambda)[A,B]^{\ot2}} \notag\\
&=&\frac1{2d(d^2+1)}\Br{\Tr{A^2B^2} - \Tr{ABAB}},\notag
 \\&=&\frac1{2d(d^2+1)}\Br{\Omega_7(A,B)-\Omega_8(A,B)}.
\end{eqnarray}
so we get \eqref{5}.

\subsection{Two examples in lower dimensions}\label{sect:examples}

In this section, we will present two examples in lower dimensions.
Note that the results obtained previously are live in the space of
the dimension being larger than three, as examples, we will deal
with the same problem in the 2-dimensional and 3-dimensional spaces,
respectively.

\begin{thrm}\label{th:dim=2}
For two observables $A$ and $B$ on $\complex^2$, the average of
uncertainty-product taken over the whole set of all density matrices
$\density{\complex^2}$ is given by
\begin{eqnarray}
&&\int\Delta A(\rho)^2\cdot\Delta B(\rho)^2\dif\mu_{\rH\rS}(\rho)\notag\\
&&=\frac{2}{105}\Omega_1-\frac{2}{35}\Omega_2+\frac{4}{105}\Omega_3+\frac{29}{210}\Omega_4+\frac{1}{105}\Omega_5
-\frac{1}{21}\Omega_6+\frac{3}{70}\Omega_7+\frac{1}{210}\Omega_8.
\end{eqnarray}
Moreover, we have
\begin{eqnarray}
&&\int
\dif\mu_{\rH\rS}(\rho)\Br{\Pa{\langle\set{A,B}\rangle_\rho-\langle
A\rangle_\rho\langle B\rangle_\rho}^2 +\langle
[A,B]\rangle^2_\rho}\notag\\
&&=\frac{2}{105}\Omega_1+\frac{1}{105}\Omega_2-\frac{2}{21}\Omega_3+\frac{1}{210}\Omega_4+\frac{1}{7}\Omega_5
-\frac{1}{21}\Omega_6+\frac{8}{105}\Omega_7-\frac{1}{35}\Omega_8.
\end{eqnarray}
\end{thrm}

\begin{proof}
For $d=2$, we have
\begin{eqnarray}
\langle t_2\rangle_2=\frac{4}{5}, ~~\langle
t_3\rangle_2=\frac{7}{10}, ~~\langle t_4\rangle_2=\frac{22}{35},~~
\langle t_2^2\rangle_2=\frac{23}{35}.
\end{eqnarray}
Hence,
\begin{eqnarray}
&&\langle \Delta_2^{(2)}\rangle_2=\frac{3}{10},~ \langle \Delta_2^{(1,1)}\rangle_2=\frac{1}{10},\\
&&\langle \Delta_3^{(3)}\rangle_2=\frac{1}{5},~ \langle
\Delta_3^{(2,1)}\rangle_2=\frac{1}{20},
\\
&&\langle \Delta_4^{(4)}\rangle_2=\frac{1}{7},~ \langle
\Delta_4^{(3,1)}\rangle_2=\frac{1}{35}, ~ \langle
\Delta_4^{(2,2)}\rangle_2=\frac{1}{70}.
\end{eqnarray}
Since
\begin{eqnarray}
\int_{\density{\cH_2}} \Tr{\Br{A^2\ot
B^2}\sE_2(\Lambda)}\dif\nu(\Lambda)&=&\frac{3}{20}(\Omega_4+\Omega_7)+\frac{1}{20}(\Omega_4-\Omega_7)
\notag\\
&=&\frac{1}{5}\Omega_4+\frac{1}{10}\Omega_7,
\end{eqnarray}
it follows that
\begin{eqnarray}
&&\int_{\density{\cH_2}} \Tr{\Br{A^{\otimes2}\ot
B^{\otimes2}}\sE_4(\Lambda)}\dif\nu(\Lambda)\notag\\
&&=\frac{1}{7}\times\frac{1}{24}
\Pa{\Omega_1+\Omega_2+4\Omega_3+\Omega_4+2\Omega_5+4\Omega_6+4\Omega_7+2\Omega_8}
\notag\\&&+\frac{1}{35}\times\frac{1}{8}\Pa{3\Omega_1+\Omega_2+4\Omega_3-\Omega_4-2\Omega_5-4\Omega_7-2\Omega_8}
\notag\\&&+\frac{1}{70}\times\frac{1}{12}\Pa{2\Omega_1+2\Omega_4+4\Omega_5-4\Omega_6}
\notag\\&&=\frac{2}{105}\Omega_1+\frac{1}{105}\Omega_2+\frac{4}{105}\Omega_3+\frac{1}{210}\Omega_4+\frac{1}{105}\Omega_5
+\frac{2}{105}\Omega_6+\frac{1}{105}\Omega_7+\frac{1}{210}\Omega_8,
\end{eqnarray}
and
\begin{eqnarray}
&&\int_{\density{\cH_2}} \Tr{\Br{A^2\ot B^{\otimes2}+B^2\ot
A^{\otimes2}}\sE_3(\Lambda)}\dif\nu(\Lambda)\notag\\
&&=\frac{1}{5}\times\frac{1}{6}
\Pa{\Omega_2+2\Omega_4+2\Omega_6+4\Omega_7}
+\frac{1}{20}\times\frac{2}{3}\Pa{\Omega_2-2\Omega_7}
\notag\\
&&=\frac{1}{15}\Omega_2+\frac{1}{15}\Omega_4+\frac{1}{15}\Omega_6+\frac{1}{15}\Omega_7,
\end{eqnarray}
then by ~\eqref{6}, we get
\begin{eqnarray}
&&\int_{\density{\cH_2}} \Delta A(\rho)^2\cdot\Delta
B(\rho)^2\dif\mu_{\rH\rS}(\rho)\notag\\
&&=\frac{2}{105}\Omega_1-\frac{2}{35}\Omega_2+\frac{4}{105}\Omega_3+\frac{29}{210}\Omega_4+\frac{1}{105}\Omega_5
-\frac{1}{21}\Omega_6+\frac{3}{70}\Omega_7+\frac{1}{210}\Omega_8.
\end{eqnarray}
Since
\begin{eqnarray}
\int_{\density{\cH_2}} \Tr{\{A,B\}^{\otimes
2}\sE_2(\Lambda)}\dif\nu(\Lambda)
=\frac{1}{5}\Omega_5+\frac{1}{20}\Pa{\Omega_7+\Omega_8},
\end{eqnarray}
\begin{eqnarray}
\int_{\density{\cH_2}} \Tr{\Br{\{A,B\}\otimes A\otimes
B}\sE_3(\Lambda)}\dif\nu(\Lambda)
=\frac{1}{15}\Omega_3+\frac{1}{30}\Pa{\Omega_5+\Omega_6}+\frac{1}{60}\Pa{\Omega_7+\Omega_8}
\end{eqnarray}
and
\begin{eqnarray}
\int_{\density{\cH_2}} \Tr{[A,B]^{\otimes
2}\sE_2(\Lambda)}\dif\nu(\Lambda)
=\frac{1}{20}\Pa{\Omega_7-\Omega_8},
\end{eqnarray}
then by \eqref{12}, we get
\begin{eqnarray}
&&\int_{\density{\cH_2}}
\dif\mu_{\rH\rS}(\rho)\Br{\Pa{\langle\set{A,B}\rangle_\rho-\langle
A\rangle_\rho\langle B\rangle_\rho}^2 +\langle
[A,B]\rangle^2_\rho}\notag\\
&&=\frac{2}{105}\Omega_1+\frac{1}{105}\Omega_2-\frac{2}{21}\Omega_3+\frac{1}{210}\Omega_4+\frac{1}{7}\Omega_5
-\frac{1}{21}\Omega_6+\frac{8}{105}\Omega_7-\frac{1}{35}\Omega_8.
\end{eqnarray}
We are done.
\end{proof}

\begin{thrm}\label{th:dim=3}
For two observables $A$ and $B$ on $\complex^3$, the average of
uncertainty-product taken over the whole set of all density matrices
$\density{\complex^3}$ is given by
\begin{eqnarray}
&&\int\Delta A(\rho)^2\cdot\Delta B(\rho)^2\dif\mu_{\rH\rS}(\rho)\notag\\
&&=\frac{3}{440}\Omega_1-\frac{1}{40}\Omega_2+\frac{1}{110}\Omega_3+\frac{109}{1032}\Omega_4+\frac{1}{660}\Omega_5
-\frac{1}{66}\Omega_6+\frac{1}{45}\Omega_7+\frac{1}{1980}\Omega_8.
\end{eqnarray}
Moreover, we have
\begin{eqnarray}
&&\int
\dif\mu_{\rH\rS}(\rho)\Br{\Pa{\langle\set{A,B}\rangle_\rho-\langle
A\rangle_\rho\langle B\rangle_\rho}^2 +\langle
[A,B]\rangle^2_\rho}\notag\\
&&=\frac{3}{440}\Omega_1+\frac{1}{440}\Omega_2-\frac{1}{22}\Omega_3+\frac{1}{1320}\Omega_4+\frac{1}{12}\Omega_5
-\frac{1}{66}\Omega_6+\frac{14}{495}\Omega_7-\frac{1}{180}\Omega_8.
\end{eqnarray}
\end{thrm}

\begin{proof}
For $d=3$, we have
\begin{eqnarray}
\langle
t_2\rangle_3=\frac{3}{5}, ~~\langle t_3\rangle_3=\frac{23}{55},
~~\langle t_4\rangle_3=\frac{17}{55},~~ \langle
t_2^2\rangle_3=\frac{61}{165}.
\end{eqnarray}
Hence,
\begin{eqnarray}
&&\langle \Delta_2^{(2)}\rangle_3=\frac{2}{15},~ \langle \Delta_2^{(1,1)}\rangle_3=\frac{1}{15},\\
&&\langle \Delta_3^{(3)}\rangle_3=\frac{2}{33},~ \langle
\Delta_3^{(2,1)}\rangle_3=\frac{4}{165},~ \langle
\Delta_3^{(1,1,1)}\rangle_3=\frac{1}{165}
\\
&&\langle \Delta_4^{(4)}\rangle_3=\frac{1}{33},~ \langle
\Delta_4^{(3,1)}\rangle_3=\frac{1}{99}, ~ \langle
\Delta_4^{(2,2)}\rangle_3=\frac{1}{165}~ \langle
\Delta_4^{(2,1,1)}\rangle_3=\frac{1}{495}.
\end{eqnarray}
Since
\begin{eqnarray}
\int_{\density{\cH_3}} \Tr{\Br{A^2\ot
B^2}\sE_2(\Lambda)}\dif\nu(\Lambda)&=&\frac{1}{15}(\Omega_4+\Omega_7)+\frac{1}{30}(\Omega_4-\Omega_7)
\\&=&\frac{1}{10}\Omega_4+\frac{1}{30}\Omega_7,
\end{eqnarray}
\begin{eqnarray}
&&\int_{\density{\cH_3}} \Tr{\Br{A^{\otimes2}\ot
B^{\otimes2}}\sE_4(\Lambda)}\dif\nu(\Lambda)\notag\\
&&=\frac{1}{33}\times\frac{1}{24}
\Pa{\Omega_1+\Omega_2+4\Omega_3+\Omega_4+2\Omega_5+4\Omega_6+4\Omega_7+2\Omega_8}
\notag\\&&+\frac{1}{99}\times\frac{1}{8}\Pa{3\Omega_1+\Omega_2+4\Omega_3-\Omega_4-2\Omega_5-4\Omega_7-2\Omega_8}
\notag\\&&+\frac{1}{165}\times\frac{1}{12}\Pa{2\Omega_1+2\Omega_4+4\Omega_5-4\Omega_6}
\notag\\&&+\frac{1}{495}\times\frac{1}{8}\Pa{3\Omega_1-\Omega_2-4\Omega_3-\Omega_4-2\Omega_5+4\Omega_7+2\Omega_8}
\notag\\&&=\frac{3}{440}\Omega_1+\frac{1}{440}\Omega_2+\frac{1}{110}\Omega_3+\frac{1}{1320}\Omega_4+\frac{1}{660}\Omega_5
+\frac{1}{330}\Omega_6+\frac{1}{990}\Omega_7+\frac{1}{1980}\Omega_8
\end{eqnarray}
and
\begin{eqnarray}
&&\int_{\density{\cH_3}} \Tr{\Br{A^2\ot B^{\otimes2}+B^2\ot
A^{\otimes2}}\sE_3(\Lambda)}\dif\nu(\Lambda)\notag\\
&&=\frac{2}{33}\times\frac{1}{6}
\Pa{\Omega_2+2\Omega_4+2\Omega_6+4\Omega_7}
+\frac{4}{165}\times\frac{2}{3}\Pa{\Omega_2-2\Omega_7}
\notag\\&&+\frac{1}{165}\times\frac{1}{6}
\Pa{\Omega_2-2\Omega_4-2\Omega_6+4\Omega_7}
\notag\\&&=\frac{3}{110}\Omega_2+\frac{1}{55}\Omega_4+\frac{1}{55}\Omega_6+\frac{2}{165}\Omega_7,
\end{eqnarray}
then by \eqref{6}, we get
\begin{eqnarray}
&&\int_{\density{\cH_3}} \Delta A(\rho)^2\cdot\Delta
B(\rho)^2\dif\mu_{\rH\rS}(\rho)\notag\\
&&=\frac{3}{440}\Omega_1-\frac{1}{40}\Omega_2+\frac{1}{110}\Omega_3+\frac{109}{1032}\Omega_4+\frac{1}{660}\Omega_5
-\frac{1}{66}\Omega_6+\frac{1}{45}\Omega_7+\frac{1}{1980}\Omega_8.
\end{eqnarray}
Since
\begin{eqnarray}
\int_{\density{\cH_3}} \Tr{\{A,B\}^{\otimes
2}\sE_2(\Lambda)}\dif\nu(\Lambda)
=\frac{1}{10}\Omega_5+\frac{1}{60}\Pa{\Omega_7+\Omega_8},
\end{eqnarray}
\begin{eqnarray}
\int_{\density{\cH_3}} \Tr{\Br{\{A,B\}\otimes A\otimes B}\sE_3(\Lambda)}\dif\nu(\Lambda)
=\frac{3}{110}\Omega_3+\frac{1}{110}\Pa{\Omega_5+\Omega_6}+\frac{1}{330}\Pa{\Omega_7+\Omega_8}
\end{eqnarray}
and
\begin{eqnarray}
\int_{\density{\cH_3}} \Tr{[A,B]^{\otimes 2}\sE_2(\Lambda)}\dif\nu(\Lambda)
=\frac{1}{60}\Pa{\Omega_7-\Omega_8},
\end{eqnarray}
then by \eqref{12}, we get
\begin{eqnarray}
&&\int_{\density{\cH_3}}
\dif\mu_{\rH\rS}(\rho)\Br{\Pa{\langle\set{A,B}\rangle_\rho-\langle
A\rangle_\rho\langle B\rangle_\rho}^2 +\langle
[A,B]\rangle^2_\rho}\notag\\
&&=\frac{3}{440}\Omega_1+\frac{1}{440}\Omega_2-\frac{1}{22}\Omega_3+\frac{1}{1320}\Omega_4+\frac{1}{12}\Omega_5
-\frac{1}{66}\Omega_6+\frac{14}{495}\Omega_7-\frac{1}{180}\Omega_8.
\end{eqnarray}
This completes the proof.
\end{proof}


\end{document}